\documentclass[%
reprint,
nofootinbib,
 amsmath,amssymb,
 aps,
 pra,
]{revtex4-1}

\usepackage{mathrsfs}

\usepackage{graphicx}
\usepackage{dcolumn}
\usepackage{bm}

\usepackage{amsthm}
\usepackage[RPvoltages]{circuitikz}

\usepackage{thmtools}

\usepackage{url}
\usepackage{hyperref}
\usepackage{color}
\usepackage{bookmark}
\bookmarksetup{
  numbered, 
  open,
}

\def\({\left(}
\def\){\right)}

\newcommand{\tn}{\textnormal}

\newcommand{\R}{\mathbb{R}}

\newcommand{\abs}[1]{\left|#1\right|}

\newcommand{\ie}{\textit{i.e.}\ }
\newcommand{\vs}{\textit{vs.}\ }
\newcommand{\eg}{\textit{e.g.}\ }
\newcommand{\etc}{\textit{etc}}

\newcommand{\cf}{\textit{cf.}\ }

\newcommand{\schrod}{Schr\"odinger}

\newcommand{\x}{\mathbf{x}}

\newcommand{\p}{\mathbf{p}}

\newcommand{\n}{\mathbf{n}}

\def\sref #1{\S\ref{#1}}

\newcommand{\image}[3]{
\begin{center}
\begin{figure}[!ht]
\includegraphics[width=#2\textwidth]{#1}
\caption{\small{\label{#1}#3}}
\end{figure}
\end{center}
\vspace{-0.35in}
}

\newtheorem{theorem}{Theorem}

\newtheorem{metaprinciple}{Metaprinciple}
\makeatletter
\newcommand{\setmetaprincipletag}[1]{
  \let\oldthemetaprinciple\themetaprinciple
  \renewcommand{\themetaprinciple}{#1}
  \g@addto@macro\endmetaprinciple{
    \global\let\themetaprinciple\oldthemetaprinciple}
  }
\makeatother
\newtheorem{observation}{Observation}
\newtheorem{principle}{Principle}
\makeatletter
\newcommand{\setprincipletag}[1]{
  \let\oldtheprinciple\theprinciple
  \renewcommand{\theprinciple}{#1}
  \g@addto@macro\endprinciple{
    \global\let\theprinciple\oldtheprinciple}
  }
\makeatother
\newtheorem{problem}{Problem}

\newtheorem{claim}{Claim}
\makeatletter
\newcommand{\setclaimtag}[1]{
  \let\oldtheclaim\theclaim
  \renewcommand{\theclaim}{#1}
  \g@addto@macro\endclaim{
    \global\let\theclaim\oldtheclaim}
  }
\makeatother

\newtheorem{corollary}{Corollary}
\newtheorem{definition}{Definition}
\newtheorem{objection}{Objection}
\newtheorem{option}{Option}
\theoremstyle{remark}
\newtheorem{remark}{Remark}
\newtheorem{reply}{Reply}
\newtheorem{implication}{Implication}
\newtheorem{comment}{Comment}

\definecolor{darkblue}{RGB}{0,0,128}

\definecolor{refcolor}{RGB}{0,0,190}
\hypersetup{
    colorlinks,
    citecolor=refcolor,
    filecolor=refcolor,
    linkcolor=refcolor,
    urlcolor=refcolor
}

\begin{document}

\title{Are mental states nonlocal?}

\author{Ovidiu Cristinel Stoica}
\affiliation{
 Dept. of Theoretical Physics, NIPNE---HH, Bucharest, Romania. \\
	Email: \href{mailto:cristi.stoica@theory.nipne.ro}{cristi.stoica@theory.nipne.ro},  \href{mailto:holotronix@gmail.com}{holotronix@gmail.com}
	}%

\date{\today}

\begin{abstract}
I show that if mental states are function of physical states, then they are nonlocal, in a sense that will be explained. I argue that, if mental states are reducible to brain physics, and if they are integrated experiences, this nonlocality implies that Classical Physics is not enough, in particular the computationalist thesis does not hold. I illustrate the argument with a thought experiment. The proof of nonlocality is straightforward and general, but the result is counterintuitive, so I spend a large part of the article discussing possible objections, alternatives, and implications. I discuss the possibility that Quantum Physics allows this kind of nonlocality.
\end{abstract}


\maketitle

\section{Introduction}
\label{sintro}

Modern Science can be seen as the quest to reduce all natural phenomena to explanations based on microphysics. This position emerged because progress in Science often coincided with new reductionist explanations.

By inductive and reductionist arguments one expects that all mental processes will eventually be reduced to microphysics \cite{Churchland1981EliminativeMaterialism,sep-eliminative-materialism,Dennett1993ConsciousnessExplained,Hofstadter2007-I-am-a-strange-loop,Dennett2016Illusionism}.
There are probably as many views on what mental states are as there are authors, so I will not even attempt to elucidate this or give a definition that would satisfy everybody. I will focus on \emph{physicalist} or \emph{materialist} approaches, that claim to reduce mind to physics or matter. Various types of physicalism include \emph{functionalism}, \emph{behaviorism}, \emph{computationalism} \etc.

The general, common feature of all physicalist approaches, can be summarized as the following:
\setclaimtag{R}
\begin{claim}[of physicalism]
\label{claim:mental2physical-reduction}
Mental processes reduce to and are determined by brain physics.\footnote{Throughout this paper I will approximate the physical support of the mind with the brain, although it is known that other parts of the body play important role, \eg the \emph{enteric nervous system}. The position of \emph{externalism} with regard to mental content  consider that mental states require the environment too (see Reply \ref{reply:PhysicalMentalCorrespondence}), so in this case we can include it as well.}
\end{claim}

The apparent lack of essential contribution from quantum effects in the brain's functionality motivated many researchers
 to consider as true the following strengthening of Claim \ref{claim:mental2physical-reduction},
\setclaimtag{R${}^{+}$}
\begin{claim}[of classical physicalism]
\label{claim:mental-classical}
Mental processes reduce to Classical Physics (in the sense that Quantum Physics is not relevant for the mental processes).
\end{claim}

Logical, but also informal thinking, as well as the functional description of the mental processes, can be modeled computationally, as information processing, neural network processing, machine learning \etc.
For this reason, probably most researchers take as true an even stronger claim \cite{MccullochPitts1943LogicalCalculusOfTheIdeasImmanentInNervousActivity,Turing1950ComputingAndIntelligence,Putnam1967Computationalism,Fodor1975LanguageOfThought,Tegmark2014OurMathematicalUniverse,RussellNorvig2016ArtificialIntelligence,Dennett2016Illusionism,sep-eliminative-materialism,sep-computational-mind},
\setclaimtag{R${}^{++}$}
\begin{claim}[of classical computationalism]
\label{claim:mental-classical-computer}
Mental processes reduce to classical computing.\footnote{
While it is possible to computationally simulate with any desired approximation both classical and quantum systems, Claims \ref{claim:mental-classical} and \ref{claim:mental-classical-computer} are not considered by all researchers to be the same, some insisting that the biological substrate is essential for consciousness \cite{Searle1992RediscoveryOfMind}.
}
\end{claim}

In this article, I argue that mental states have a property that is not present in the physical states supposed to ground them.
The difference is that mental states are integrated, or unified experiences, in a sense in which classical physical states are not.

I will show that if we assume physicalism and that mental states are \emph{integrated experiences} (\cf Observation \ref{ppIntegration}), then mental states are nonlocal, in the sense that they depend instantaneously on spacelike separated events. Since this nonlocality seems different from other forms of nonlocality encountered in physics, I will call it \emph{O-nonlocality} (\cf Definition \ref{def:nonlocality-O}).

We will see that O-nonlocality seems to be in tension with relativistic locality, which grounds the very causality of physical processes.
I will argue that O-nonlocality is not supported by Classical Physics, contradicting Claims \ref{claim:mental-classical} and \ref{claim:mental-classical-computer}.
Therefore, this is a problem that has to be addressed by any physicalist program that aims to include the mental processes.

O-nonlocality suggests a sense in which a mental state is more than the sum of its parts, if these parts are the local physical quantities at each point in the relevant region of space. This may seem as preventing the reductionism of mental processes to physical ones. On the other hand, O-nonlocality is strikingly similar to \emph{quantum entanglement} (entanglement is not nonlocality \emph{per se}, but it encodes nonlocal correlations that can be objectively accessed only by measurements). Therefore, O-nonlocality may still be consistent with the physicalist Claim \ref{claim:mental2physical-reduction}, if we allow for quantum effects.

While the proof is very simple, even tautological, its implications may be counterintuitive, so I start by describing a thought experiment that may bring more intuition in Sec. \sref{s:thought-experiment}.
The more formal proof, given in Sec. \sref{s:proof}, is independent of the fundamental theory of the microphysical level.
In Sec. \sref{s:quantum} I argue that the best option is Quantum Physics, and analyze the soundness of this identification. 
In Sec. \sref{s:implications} I discuss possible implications of these results.
In Appendix \sref{a:objections} I try to anticipate and address possible objections.
In Appendix \sref{a:options} I explore possible alternative ways out.

\section{The thought experiment}
\label{s:thought-experiment}

The proof, given in Sec. \sref{s:proof}, is pretty straightforward. But since its conclusion may be counterintuitive to the reader trained to project mental states on computer states or physical states in general, let us first consider the argument in a thought experiment that should make its essential points clearer, and bring an intuition of what kind of nonlocality I am talking about. For this, I will amplify the elements of the argument to a cosmic scale.

In \sref{s:thought-experiment-digital} I discuss a version of the thought experiment which assumes Claim \ref{claim:mental-classical-computer}.
Claim \ref{claim:mental-classical-computer} is stronger than Claim \ref{claim:mental-classical}, but it simplifies the argument showing that mental states involve a certain kind of nonlocality.
Another reason to focus separately on Claim \ref{claim:mental-classical-computer} is that it is supported by various researchers. 

Then, in \sref{s:thought-experiment-analog}, I propose and discuss a version of the thought experiment which involves a biological brain, satisfying the more general Claim \ref{claim:mental-classical}.

\subsection{The cosmic computer thought experiment}
\label{s:thought-experiment-digital}

Consider a hypothetical classical computer supporting mental states.
Let us spread its components throughout our galaxy, or even across more galaxies, by placing each one of them on a separate space station orbiting a different star. Suppose we also arrange that the space stations storing these small parts exchange electromagnetic signals across the galaxy, to ensure the flow of information necessary for the functioning of the computer.

It can be arranged, in principle, that our cosmic computer is fragmented into very small parts, storing or processing the smallest numbers of bits each. Two bits on each star should suffice, because any logical gate can be made of two-bit gates, as explained in Reply \ref{reply:gates}.

The mental states of this cosmic mind change really slowly, each bit requiring thousands or maybe millions of years to be processed. 
But if we assume that a computer can have mental states, then it should be possible for this cosmic contraption to have mental states too. There should be a mapping of the computer's states to mental states.

We can even replace the bits with humans who represent the bit by wearing or not a hat, and passing the hat from generation to generation to compensate for their too short lives compared to the duration of the processing. We can obtain in this way cosmic versions of the thought experiments from the \emph{Chinese nation argument} \cite{sep-chinese-room,NedBlock1978TroublesWithFunctionalism} or of Searle's \emph{Chinese room} argument \cite{Searle1980ChineseRoom}. 

But, compared to these well-known arguments, there is a plot twist here: \emph{the aim is to show that mental states are nonlocal in a certain sense}. And if this may not be so obvious when we are talking about the brain, which appears to our senses pretty much localized in space, this cosmic version should, hopefully, make this nonlocality evident even to the most skeptic readers, because the bits supporting the mental states are isolated and separated by spacelike intervals.

Several questions become natural at this point. Does this cosmic mental process take place continuously, or only when one of the bits flips its state? Are there mental states associated to the thousands of years of stagnation of the bits, time in which electromagnetic signals travel across the universe from one star to another?

Regardless of the answer to these questions, the corresponding mental states are nevertheless nonlocal, in the sense that any such state depends on bits located across very distant stars, maybe galaxies apart from one another, separated by spacelike intervals.
Let us make these problems more obvious.

Classical serial computers process one bit at a time (see Reply \ref{reply:enforced-locality}). Even if we make the computer parallel, there still is a central unit that breaks the task into smaller tasks, and then centralizes the results, and it does so one bit at a time. 
If the mental states are supported exclusively when the bits are flipped, and if we arrange that the computer processes one bit at a time, it seems unlikely for a bit to support complex mental states like happiness or sadness. So maybe the states of the other bits matter too, even if they are not flipped right at that time. But this raises other problems.
\begin{problem}
\label{pb:bits-matter-only}
If the mental states are supported exclusively by the states of the bits, then how is the
instantaneous
mental state grounded in the configuration of the bits spread across the galaxy, separated by spacelike intervals? How does this grounding manage to yield a unified, integrated,
instantaneous
 mental state, from these apparently disorganized zeroes and ones?
\end{problem}

\begin{problem}
\label{pb:bits-matter-only-discrimination}
If there are other space stations containing bits that are not part of the same computer, then how is the right subset of these bits selected out of them to ground the 
instantaneous
mental state of our computer, and yield integrated experience? What discriminates which subset of these bits belong to the computer that supports this mental state and which do not?
\end{problem}

We may hope that if we assume that not only the bit configuration supports mental states, but also the signals traveling on their way between different space stations, this would help discriminate the bits from the same configuration, and also help the mental state to integrate the configuration.
But Problems \ref{pb:bits-matter-only} and \ref{pb:bits-matter-only-discrimination} remain, they only extend to include the traveling signals.

On top of these problems, we have another one. 

\begin{problem}
\label{pb:camouflage}
There is no rule that the bits are stored in a certain way on the space stations. One can invent complicated contraptions to store the bits as coins on a table, or in the position of any object on the space station, or in whatever states any kind of object can have. Important is to have a way to read and write the bits by analyzing and rearranging these objects, and there are infinitely many different ways to do this. How is then a mental state able to emerge out of such a configuration camouflaged in such a way in the state of the universe? If we insist on supporting Claim \ref{claim:mental-classical-computer}, then there seem to be only a way out: accept that more possible subsets of the configuration of the universe support independent potential mental states!
Are different possible subsets of the state of the universe able to support,
provided that their instantaneous configuration can be identified with the state of an appropriate machine,
independent mental states?
\end{problem}

These Problems are not necessarily of computationalism, but of classical computationalism under the assumption of locality. These arguments do not necessarily refute quantum computationalism, only the classical one from Claim \ref{claim:mental-classical-computer}.

If we assume that mental states exist and are grounded at every instant in the configuration of the physical state, then these questions suggest that there must be something nonlocal about the mental states.
The reader who denies that mental states exist instantaneously can try to escape the necessity of nonlocality, by assuming that they have an extended duration. This possibility will be discussed, particularly in Objection \ref{obj:nonlocality-in-time} and Implication \ref{implication:relativity}.
The reader can also deny that mental states even exist at all, and conclude that therefore the problems mentioned here don't exist. This will come in Option \ref{opt:no-mental-states}. Another way out is the position that one should not discuss about what we cannot objectively observe, as in Option \ref{opt:instrumentalism}. There is nothing I can do about these positions. But for the reader who takes mental states as real and grounded in the physical state at every instant, I hope this thought experiment
suggest that \emph{if we claim that 
instantaneous
mental processes are determined by the computational processes of our cosmic computer, then there is an essential sense in which these mental states have to be nonlocal}. This is what I mean by \emph{O-nonlocality}.
And this thought experiment suggests that classicality cannot accommodate this O-nonlocality.

\subsection{The biological brain thought experiment}
\label{s:thought-experiment-analog}

There are several significant differences between how brains and computers work. Brains employ \emph{neural networks}, but at least \emph{artificial neural networks} can be simulated on Turing-type computers. The brain seems to be analog, but we can make the case that even if it is analog, it can be approximated with a digital one to any desired degree, because what matters is the distinguishability of the states, which is limited. In addition, Quantum Physics implies that a localized system like the brain can only have a discrete, or even a finite number of distinguishable states, and this is true even if the brain's relevant functionality is quasiclassical. These arguments support the idea that the functionality of the brain can be simulated on a computer how closely we desire, a version of the \emph{Church-Turing thesis}. And if we believe that only the behavior matters, since the behavior of the mind can be simulated, the mind should be like a Turing machine.

However, maybe not only the computational or behavioral aspects of the brain are relevant for the mental processes, but also the fact that it is biological, its \emph{material substrate}. This is could make Claim \ref{claim:mental-classical-computer} to be stronger than Claim \ref{claim:mental-classical}. So we need to see if we can make the thought experiment from \sref{s:thought-experiment-digital} more general, by including the material substrate.

A possible way to adapt the cosmic computer thought experiment to the biological brain is to try to divide the brain into the smallest units for which the substrate is important. Maybe it is possible to divide it into neurons, or maybe into smaller parts of the neurons, which are made of the relevant substrate. Then, it is conceivable that we can replace the connections that allow signals to be exchanged between these parts, and the parts whose substrate is irrelevant, with other mechanisms that allow us to separate them throughout the galaxy, in a way similar to the cosmic computer thought experiment. But this requires an understanding of the brain that we do not currently have.

Another way, which will be used here, is to zoom-in, in a way similar to \emph{Lebniz's Mill argument}. Leibniz imagined zooming-in a brain until one can walk inside of it and see its machinery like inside a mill \cite{Leibniz1989Monadology}. Imagine that we zoomed-in, so that the brain appears as large as a galaxy. This would make its constitutive atoms of planetary sizes. Now we can see the problem, it is similar to the case of the cosmic computer thought experiment. While we no longer have bits, we have instead the states of the atoms, and their relative configurations. But the states of an atom are discrete, and the possible ways atoms combine into molecules are discrete too. So the biological brain version of the thought experiment is similar to the cosmic computer one in its essential aspects, except that in addition one assumes that the substrate is relevant too. In this case, the substrate is provided by the atoms and their arrangements, or even by the configurations of the electrons and the nucleons composing the atoms. At this level where Atomic and Molecular Physics and Chemistry become relevant, Quantum Physics becomes relevant too. But if we want to abide to Claim \ref{claim:mental-classical}, we have to assume that only the classical limit of the configurations of atoms and molecules is relevant.

At this point, we can see that even if we assume that the substrate is important for the mental states, this cannot avoid the problems that arise if we assume locality. In particular, this thought experiment reveals problems similar to Problems \ref{pb:bits-matter-only} and \ref{pb:bits-matter-only-discrimination}, with the amendment that instead of states of bits we are talking about quasiclassical approximations of states consisting of atoms and molecules.

Note that zooming-in to increase the apparent size of the brain did not actually change the way it works. The purpose of zooming-in was not to ``destroy'' the mental states, but only to make clearer the problem introduced by locality and spacelike separation between the atoms. So the O-nonlocality of mental states is manifest in the case of the biological brain too.
Not the classical computationalist Claim \ref{claim:mental-classical-computer} was the cause of the problems, but the assumption of locality.

\section{The proofs}
\label{s:proof}

I will now formalize the argument.
To allow the reader to see where possible objections fit better, I will try to identify and highlight all of the assumptions and steps in the proofs and the arguments, no matter how obvious they may be.
Hopefully I will also correct some assumptions we may be tempted to make.

For theoretical-independence and generality (to protect the arguments from objections appealing to unknown physics), I rely in my arguments only on a common structure of most, if not all theories discovered so far: the representation of \emph{processes as temporal successions of physical states}. 
This mathematical structure is called \emph{dynamical system}, and it is common to classical and quantum, nonrelativistic and relativistic (even General Relativity \cite{adm2008admRepublication}), continuous and discrete and in fact to all more or less successful theories conceived so far.

Very generally, but not more formally than needed for the argument, a dynamical system can be defined as
\begin{enumerate}
	\item A \emph{state space} $S$, which is the set of all possible states of the system allowed by the theory.
	\item A set $H$ of subsets of $S$, called \emph{histories}, so that for each history $h\in H$ there is a surjective function $t:T\to h$ from a totally ordered set $T$ to $h$.
\end{enumerate}

The totally ordered set $T$ represents the \emph{time}\footnote{In general, $T$ also has a group or semigroup structure compatible with the order relation, but we do not need it here.}. Each history represents a \emph{process}, and specifies the state of the system for each instant $t\in T$.
The dynamical system can be given by specifying a \emph{dynamical law} governing the \emph{temporal evolution of states}, which is equivalent to specifying all possible histories\footnote{
In Newtonian Mechanics a state consists of the positions and momenta of all particles. The possible histories are summarized by \emph{Hamilton's equations}. In Standard Quantum Physics, a state is represented by a vector from a Hilbert space. The histories are given by \emph{{\schrod}'s equation}, interrupted during measurements by jumps whose probabilities are given by \emph{Born's rule}.}.

We do not know the final theory that can completely describe all natural phenomena, but all we know so far is consistent with the following
\setmetaprincipletag{DS}
\begin{metaprinciple}[dynamical system]
\label{ppDynamicalSystem}
The physical world is described by a dynamical system.
\end{metaprinciple}

A \emph{partial function} from a set $A$ to a set $B$ is a function $f:A'\to B$, where $A'\subseteq A$. For simplicity, I often called the partial function a function on $A$, even if it is partial, when it is clear from the context what I mean\footnote{The more exigent reader can interpret $f$ as a function on $A$, where $f(a)=d$ for $a\in A\setminus A'$, $d\notin B$ being a ``dummy'' element.}.

\begin{definition}[Properties]
\label{def:properties}
A \emph{property} of a system or a subsystem is a (partial) function of the state of the system, valued in some set (usually the real numbers $\R$).
\end{definition}
For example, in Classical Mechanics, properties are functions on the state space. In Quantum Physics, they are often not functions, but partial functions. For example, most state vectors don't have a definite position or momentum.

Because any process is a succession of states, Claim \ref{claim:mental2physical-reduction} implies that \emph{mental states are properties of physical states}. This can be formalized as a Principle (Fig. \ref{mental-states-correlates.png}):
\setprincipletag{PM}
\begin{principle}[physical-mental correspondence]
\label{ppPhysicalMentalCorrespondence}
A function $\Psi$ associates to (some of the) physical states the corresponding mental states that they determine,
\begin{equation}
\label{eq:psi-mental-states}
\tn{\texttt{mental\_state}}=\Psi\(\tn{\texttt{physical\_state}}\)
\end{equation}
\end{principle}

\image{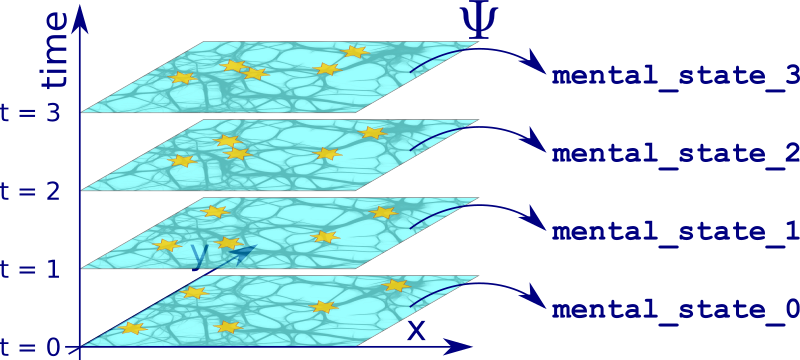}{0.42}{Schematic representation of a network of neurons whose state evolves in time, determining a succession of mental states via the function $\Psi$ from Principle \ref{ppPhysicalMentalCorrespondence}.}

I detail how the major physicalist positions about the mental processes satisfy Principle \ref{ppPhysicalMentalCorrespondence} in the Reply to Objection \ref{obj:PhysicalMentalCorrespondence}.

A very trivial assumption is the following
\begin{observation}
\label{ppMentalExtended}
The physical states underlying mental states are extended in space.
\end{observation}

This was already observed by William James, who wrote ``[t]here is no cell or group of cells in the brain of such anatomical or functional preeminence as to appear to be the keystone or center of gravity of the whole system'' (\cite{WilliamJames1890ThePrinciplesOfPsychology}, p. 180). Likewise, Dennett (\cite{Dennett1993ConsciousnessExplained}, p. 102) states that ``[t]here is no single point in the brain where all information funnels in'', and ``[t]he idea of a special center in the brain is the most tenacious bad idea bedeviling our attempts to think about consciousness''.

This can be related to Damasio's statement ``[w]hat we experience as mental states corresponds not just to activity in a discrete brain area but rather to the result of massive recursive signaling involving multiple regions'' \cite{Damasio2012SelfComesToMind}.
Damasio expresses the complexity of mental states
\begin{quote}
Even with the help of neuroscience techniques more powerful than are available
today, we are unlikely ever to chart the full scope of neural phenomena
associated with a mental state, even a simple one.
\end{quote}

I will discuss more evidence for this observation in the replies to the objections in \sref{a:enforced-locality}.

I now introduce some definitions that will be used in the argument.

\begin{definition}[Records]
\label{def:records}
\emph{Recording} is the act of changing the state of a system, called \emph{recorder}, by bringing it, according to the physical laws, into a state that corresponds to the state of another system, called \emph{recorded system}, so that distinct states of the recorder correspond to distinct states of the recorded system. The resulting distinct possible states of the recorder are called \emph{records}. A recorder can be recorded in its turn.
When the records can be interpreted numerically, the act of recording is called \emph{measurement}, and the recorder is called \emph{measuring apparatus}.
\end{definition}

\begin{definition}[Direct records]
\label{def:direct-record}
When the recorder and the recorded system coincide or overlap, or are the coarse grainings of more fundamental states that coincide or overlap, we say that the recorder has \emph{direct access} to the recorded system.
Otherwise we say that the recorder has \emph{indirect access} to the recorded system.
\end{definition}

Since mental states can discriminate between various observables, Definition \ref{def:records} and Principle \ref{ppPhysicalMentalCorrespondence} justify the following definitions.
\begin{definition}[Observer]
\label{def:observer}
When the physical state underlying a mental state can distinguish the properties of a system as in Definition \ref{def:records}, it acts like a recorder. If in addition $\Psi$ from eq. \eqref{eq:psi-mental-states} has different values for these properties, is also acts like a recorder, which we call \emph{observer}, and we call its act of recording \emph{observation}.
Records that can be distinguished by an observer are called \emph{observables}.
\end{definition}

\begin{observation}[Direct observations]
\label{ppDirectObservation}
Since by Definition \ref{def:observer} observations are those recordings made by mental states, the only \emph{directly observable} properties are observable properties of the physical system underlying the very mental state observing them. In other words, properties distinguished by the function $\Psi$ from eq. \eqref{eq:psi-mental-states}.
\end{observation}

All observations, including the scientific ones, are made by recording systems that have mental states, called observers. While the position that there is something fundamental about the consciousness of the observers is considered controversial, the fact that they have memories in which they record information about the data, and that appeal to each other's memories, even when conducting scientific research, is generally accepted.
However, the role of Definition \ref{def:observer} is not to introduce mental states as something fundamental, but merely to distinguish those recording systems endowed with mental states.

But what is the minimal feature that distinguishes mental states among other recorders? I propose that this is the \emph{Integration Property} defined below, which is the central assumption that I require the reader to accept about mental states.

\begin{observation}[Integration Property]
\label{ppIntegration}
Conscious mental states are integrated experiences. 
\end{observation}

Here by ``integrated'' I mean the simple fact that, despite Observation \ref{ppMentalExtended}, the mental state experience is unified, even if only partially, rather than completely scattered in disconnected elementary parts.
What is a combination or integration of multiple factors, appears to us as a single, unified experience, which is the very mental state. A conscious mental state is the awareness of a combination of the multiple facets, of all of the structural parts making that very mental state.
For a particular mental state there may be an experience that merely looks like a unified self, which is largely an illusion and \emph{folk psychology}, see \eg \cite{Churchland1981EliminativeMaterialism,Dennett1993ConsciousnessExplained,Damasio2012SelfComesToMind}, but the experience of that illusion is nevertheless integrated, at least to some degree. Many illusions are involved here, but even these illusions are an integrated experience. If only a subset of the totality of the brain state can be experienced as a whole, like a \emph{draft} if the reader wishes \cite{Dennett1993ConsciousnessExplained}, I refer only to such a subset as being integrated.
I do not assume the widely discredited \emph{Cartesian Theater} or a \emph{homunculus}, as Dennett says people do by default \cite{Dennett2016Illusionism,Dennett2017FromBacteriaToBachAndBack},
although I suspect that in many cases people don't really assume such a homunculus, but rather that their experiences are integrated.
It can be objected that the only evidence for this integration comes through introspection, known in Psychology and Neuroscience to be an unreliable tool due to its subjectivity, while science is about objectively verifiable evidence. But the only lesson I am taking here from introspection and self-inquiry is that the mental states exist and their experience is integrated, even if they are not what they seem to us to be. The researchers who consider that the only evidence that counts scientifically should be objectively verifiable, and for whom even subjective experiments that can be repeated independently are invalid because each instance of these experiments is a different mind, can take as an exit route out of my arguments the alternative Option \ref{opt:instrumentalism} or Option \ref{opt:no-mental-states} from Appendix \sref{a:options}.

I do not require from the reader any ontological commitments or a particular position with respect to the 
\emph{Hard Problem of consciousness} \cite{Chalmers1995HardProblem}. 
According to Chalmers \cite{Chalmers1996ConsciousMindSearchFundamentalTheory}, ``a mental state is conscious if there is something it is like to be in that mental state. [...] This is the really hard part of the mind-body problem.''
Although some authors deny the very existence of such a problem \cite{Dennett2016Illusionism}, I tried to argue elsewhere \cite{Sto2020NegativeWayToSentience} that the problem is real and hard.
However, I consider that the arguments in this paper should be seen as independent of the Hard Problem, so I avoided discussing it in this paper, and I only assume that mental states have the {Integration Property}.

We can now define the type of nonlocality that we will prove to be associated to mental states.
\begin{definition}[O-nonlocality]
\label{def:nonlocality-O}
A property of a system is called \emph{O-nonlocal} if it is directly observable (as in Definition \ref{ppDirectObservation}) and depends instantaneously on spacelike separated events.
\end{definition}

\begin{remark}
\label{rem:separation-integration}
Let us consider a mental state having the {Integration Property}. Suppose that this mental state involves awareness, and that its physical correlates are located in a region of space $R$. Consider two disjoint subsets $A$ and $B$ of $R$, \ie $A \cap B=\emptyset$, so that $A \cup B = R$, as in Fig. \ref{mental-states-integrated.png}. 

\image{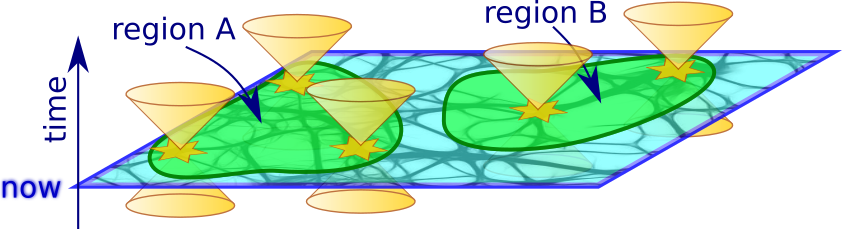}{0.45}{Illustration of the O-nonlocality associated to the {Integration Property} and Remark \ref{rem:separation-integration}. The Figure represents a mental state grounded on a physical state that has two subsystems restricted to two disjoint regions A and B of space (the green regions).}

Since the mental state is grounded on region $A\cup B$, it depends on the state of the subsystem contained in $A$, but also on the state of the subsystem contained in $B$. This dependency of the mental state on the physical correlates in these regions, \emph{which are spacelike separated}, is an essential feature of O-nonlocality. The other one is \emph{direct observability}, as in Observation \ref{ppDirectObservation}.
\end{remark}

We are now ready for the proof of O-nonlocality, and the arguments that mental states are not consistent with Claim \ref{claim:mental-classical} (and implicitly neither with Claim \ref{claim:mental-classical-computer}).

\begin{theorem}[of O-nonlocality]
\label{thm:nonlocality-O}
If Principle \ref{ppPhysicalMentalCorrespondence} is true, then mental states are O-nonlocal.
\end{theorem}
\begin{proof}
I assume as obvious that the physical states underlying mental states are extended in space (Observation \ref{ppMentalExtended}).
An immediate consequence of Metaprinciple \ref{ppDynamicalSystem} (\nameref*{ppDynamicalSystem}) -- which is assumed by Principle \ref{ppPhysicalMentalCorrespondence} -- is that a system cannot directly access other states in the state space, not even its own previous states!
It can only have \emph{indirect access}, as in Definition \ref{def:direct-record}.

\begin{observation}[Instantism]
\label{ppInstantism}
Any system can have direct access only to its present state (Fig. \ref{mental-states-access2past.png}).
\end{observation}

\image{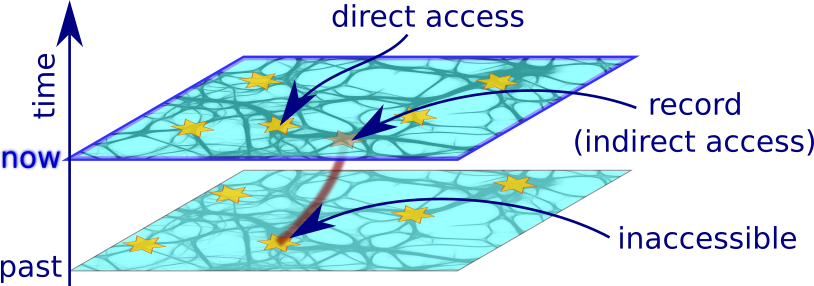}{0.45}{A physical system can access directly only its present state. Consequently, a mental state can access the past only indirectly, as records of the past physical states that exist in the present physical state underlying that mental state.}

At first sight, this is a very trivial observation. But in fact its implications are too often overlooked.
Even if it is trivially true, the reader may object that, if we can only access our present state, then how is it even possible to remember the past? I will come back to this known issue in Appendix \sref{a:objections}, Objection \ref{obj:records}.

In a local theory, an \emph{event} -- which consists of its position and moment of time -- can have direct access only to itself.
In relativistic theories, even indirect access is limited to current records of the events from the \emph{past lightcones} of the current event, as shown in Fig.  \ref{mental-states-access2past.png}. 
Simultaneous events (with respect to a reference frame) are isolated, ``blind'' to one another. 
This leads to
\begin{observation}[Locality]
\label{ppLocality}
In a local theory, any event can have direct access only to the physical quantities located at that event.
\end{observation}

Fig. \ref{mental-states-nonlocal.png} illustrates Observation \ref{ppLocality} and shows how the events composing the present state of a brain are isolated from one another in this sense.

\image{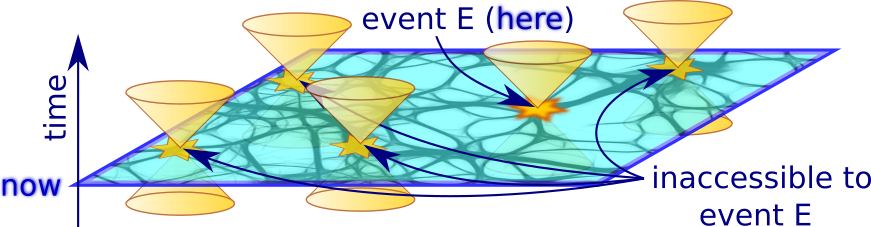}{0.45}{Illustration of Observation \ref{ppLocality}, which shows that in a local theory, any event can have direct access only to the physical quantities located at that event.}

Let us now collect all the observation. The physical state underlying a mental state is, according to Observations \ref{ppMentalExtended} and \ref{ppInstantism}, extended in space in a spacelike way. By Observation \ref{ppLocality} its parts cannot, directly or indirectly, access other parts. Due to Principle \ref{ppPhysicalMentalCorrespondence}, the mental state is a coarse graining of the underlying physical state or coincides with it (in the case when $\Psi$ is the identigy function). By Definition \ref{def:direct-record}, the mental state has direct access to its underlying physical state, and by Definition \ref{def:observer}, this is a direct observation.
Since the mental state directly observes its underlying physical state and depends instantaneously on its spacelike separated events, Definition \ref{def:nonlocality-O} is satisfied, and mental states turn out to be O-nonlocal.
\end{proof}

Some possible objection to the idea that O-nonlocality may be nonclassical are discussed in \sref{a:objections-space}.
They all point out towards the idea that O-nonlocality may be actually a very classical property, and the replies appeal to the {Integration Property} to explain the difference.

\begin{remark}
\label{rem:refute_nonlocality}
Nothing can stop the reader to simply reject the {Integration Property}. Nothing can stop the reader, if she or he wishes, to imagine that there is no connection between regions $A$ and $B$ from Remark \ref{rem:separation-integration}, or between any two distinct points from region $R$, and then to claim that we can still have experience somehow without the need of some nonclassical O-nonlocality. I think this presumed experience is in fact a projection of the experience we have when we mentally simulate, in our own mental states, such a possibility. 
\end{remark}

Since Classical Physics is local, Theorem \ref{thm:nonlocality-O} implies
\begin{corollary}
\label{thm:claims_contradiction}
Assuming Observation \ref{ppIntegration}, Claim \ref{claim:mental-classical} cannot be true (unless, of course, Classical Physics is updated with a certain nonlocality able to yield the integrated experience associated to the mental state).
\end{corollary}
\begin{proof}
I also assume Observation \ref{ppMentalExtended} as obvious. I consider each of two cases.

\emph{Case I: Metaprinciple \ref{ppDynamicalSystem} is false for the brain physics}. Since Metaprinciple \ref{ppDynamicalSystem} is true in Classical Physics, this case is inconsistent with Claim \ref{claim:mental-classical}.

\emph{Case II: Metaprinciple \ref{ppDynamicalSystem} is true}. Since Claim \ref{claim:mental-classical} implies Principle \ref{ppPhysicalMentalCorrespondence}, we can apply Theorem \ref{thm:nonlocality-O}, according to which mental states are O-nonlocal. This is again inconsistent with Claim \ref{claim:mental-classical}, which requires locality.
\end{proof}

Since Claim \ref{claim:mental-classical-computer} implies Claim \ref{claim:mental-classical}, Corollary \ref{thm:claims_contradiction} also rejects Claim \ref{claim:mental-classical-computer}.

\begin{corollary}
\label{thm:quantum_mind}
In a world where the only nonlocality is quantum, the {Integration Property} and Claim \ref{claim:mental2physical-reduction} require that any system supporting mental states has to use quantum effects.
\end{corollary}
\begin{proof}
Follows directly from Theorem \ref{thm:nonlocality-O}.
\end{proof}

The reader who wants to save Claims \ref{claim:mental-classical} and \ref{claim:mental-classical-computer} can try to reject the {Integration Property}, and assume therefore that mental states are as scattered and spacelike separated like the physical states grounding them are. Some proposed alternatives allowing this are discussed in Appendix \sref{a:options}.

\section{Can the nonlocality of mental states be quantum?}
\label{s:quantum}

In this section I discuss the possibility that the kind of nonlocality revealed in Sections \sref{s:proof} and \sref{s:thought-experiment} may be of quantum origin. 
This possibility is justified by the fact that we do not know of other nonlocal physics than that of quantum origin, see Corollary \ref{thm:quantum_mind}.

For reasons independent from Theorem \ref{thm:nonlocality-O}, the idea that Quantum Physics has something to do with consciousness is nearly as old as Quantum Physics itself, the main attractive features being 
\begin{enumerate}
	\item 
the violation of determinism by the wavefunction collapse, considered by \emph{incompatibilists} to be necessary for free-will,
	\item 
the apparent necessity of mentioning an observer even in the fundamental postulates of Standard Quantum Physics (but see \cite{Sto2020StandardQuantumMechanicsWithoutObservers} for a version of Standard Quantum Physics without observers),
	\item 
advantages of quantum computability over the classical one, by using superposition and entanglement.
\end{enumerate}

Numerous ideas connecting consciousness with Quantum Physics were put forward, in particular by London and Bauer \cite{LondonBauer1983TranslationOf1939QuantumMind}, Heisenberg \citep{Heisenberg1958PhysicsAndPhilosophy}, von Neumann \citep{vonNeumann1955MathFoundationsQM}, Wigner \citep{Wigner1961WignerFriend}, Penrose and Hameroff \citep{HameroffPenrose2017OrchORUpdatedReview}, Stapp \citep{Stapp2004QuantumTheoryOfMindBrainInterface,Stapp2015QuantumMechanicsMindBrainConnection}, and others \citep{SEP-qt-consciousness}.

The idea that quantum effects in the brain may be relevant for the mind was criticized in \cite{tegmark2000decoherenceBrain} and \cite{KochHepp2006QuantumMechanicsInTheBrain,Hepp2012CoherenceAndDecoherenceInTheBrain}. The decoherence times in a ``warm, wet and noisy'' environment like the brain was estimated to less than $10^{-21}$ seconds \cite{tegmark2000decoherenceBrain}, but in \cite{HaganHameroffTuszynski2002QuantumComputationInBrainMicrotubulesDecoherenceAndBiologicalFEasibility,Hameroff2006ConsciousnessNeurobiologyAndQuantumMechanics} it was argued that the corrected version gave $10-100$ microseconds, with a possibility of going up to $10-100$ milliseconds.
Moreover, evidence for maintained quantum coherence in a ``warm, wet and noisy'' environment was found later, although for photosynthesis \cite{EngelEtAl2007EvidenceQuantumCoherencePhotosynthesis,PanitchayangkoonEtAl2010LongLivedQuantumCoherencePhotosynthesis}. A simple ``recoherence'' mechanisms able to maintain coherence in an ``open and noisy'' quantum system like the brain was found in \cite{HartmannDurBriegel2006SteadyStateEntanglementInOpenNoisyQuantumSystems,LiParaoanu2009GenerationAndPropagationOfEntanglementInDrivenCoupledQubitSystems}. Other quantum features, expected to be associated to macroscopic quantum systems, were found in \emph{microtubules} \cite{Sahu2013MultiLevelMemorySwitchingPropertiesSingleBrainMicrotubule}.
But the jury is still out.

If we accept the conditions leading to Corollary \ref{thm:quantum_mind} from Sec. \sref{s:proof}, we have to accept that the O-nonlocality of mental states is due to Quantum Physics. Alternatively, we can seek for other explanations of O-nonlocality. We may consider new physics or even something nonphysical that does not aim to explain sentient experience, but merely O-nonlocality. But whatever that ``something nonphysical'' may be, if the way it provides O-nonlocality can be consistently described by a set of propositions, it can be described mathematically, and can be very well incorporated into physicalism \cite{Sto15c,Sto2020NegativeWayToSentience}. And the resulting physicalism will not be classical.

So let us see what quantum effects can support the O-nonlocality of mental states. Quantum nonlocality is the result of the conjunction of \emph{quantum entanglement} with the projection that seems to take place during measurements. Nonlocal correlations like those violating the Bell inequality are the consequence of the conjunction of these two features of Quantum Physics.

I will first describe a sense in which the O-nonlocality of the function $\Psi$ from eq. \eqref{eq:psi-mental-states}, required by Principle \ref{ppPhysicalMentalCorrespondence} in conjunction with the {Integration Property}, is analogous to \emph{quantum entanglement}. I will also explain that this analogy does not seem to be perfect.

\emph{Nonrelativistic Quantum Physics} is obtained from a classical dynamical system by \emph{quantization}. While the state of the classical system is a point in its state space, in the quantized theory it is a function defined on the classical state space, so a function of the possible classical states. For example, if the classical system consists of $\n$ (point-)particles, its state space is the \emph{phase space} $\R^{6\n}$, whose points are $(\x,\p)=(\x_1,\ldots,\x_\n,\p_1\ldots,\p_n)$, where $\x_j$ is the position of the $j$-th particle, and $\p_j$ its momentum. A quantum state is then represented by a real function $W$ on the phase space $\R^{6\n}$, called \emph{Wigner function} \cite{Wigner1932PhaseSpaceQM,Baker1958FormulationofQMPhaseSpace}. But the more commonly used representation is as a complex function $\psi$, called \emph{wavefunction}, defined on the classical \emph{configuration space} $\R^{3\n}$ of all possible positions $\x$ of $\n$ particles. The quantization procedure also replaces the classical evolution law with a quantum one, so the Hamilton equations satisfied by the classical system are replaced with a {\schrod} equation satisfied by $\psi$. The wavefunction $\psi$ depends on the position $\x$, but it also depends implicitly on the momentum $\p$, as we can see by applying the Fourier transform.

Entanglement is due to the fact that $\psi$ is a function of the positions of all particles, and it cannot be seen as more functions of the position of each particle $\psi_1(\x_1),\ldots,\psi_\n(\x_\n)$, except in special cases. In general, it is a linear combination of possibly infinitely many such products.

The analogy between the function $\Psi$ from eq. \eqref{eq:psi-mental-states} and the wavefunction $\psi$, or its Wigner transform $W(\psi)$, is that both depend on the arrangements of the particles in space. Let us represent this analogy in an admittedly quite vague way,
\begin{equation}
\label{eq:Psi-and-psi}
\Psi\(\tn{\texttt{physical\_state}}\) \tn{ is like }\psi\(\tn{\texttt{physical\_state}}\)
\end{equation}

This analogy suggests a possible relation between the O-nonlocality of the mental states, which we just inferred from Theorem \ref{thm:nonlocality-O}, and quantum nonlocality. But we should be careful that there are some differences. First, the argument of $\Psi$, \tn{\texttt{physical\_state}}, represents a physical state in our world, which is quantum. This state is not classical. Even if under Claim \ref{claim:mental-classical} we assume it to be quasiclassical, it is not the same as the classical state which is the argument of the wavefunction $\psi$. However, it is considered plausible that the classical states emerge approximately as a limit of the quantum states in regimes where entanglement is almost absent and the wavefunction $\psi$ is concentrated mostly around, and highly peaked at a point $\x$ in the configuration space.
This relation actually makes some sense, for the following reason. First, the classical state space $\R^{6\n}$ on which the wavefunction is defined assumes the classical particles to be pointlike. But in Quantum Field Theory (QFT), this state space has to be replaced with a space of classical fields. A way to obtain QFT is by the so-called ``second quantization'', which takes the wavefunction and quantizes it similarly to how classical systems of point-particles are quantized. This improves the analogy expressed vaguely in eq. \eqref{eq:Psi-and-psi}.

There is also a difference in the values the two functions $\Psi$ and $\psi$ can take.
The values of $\psi$ are complex numbers. The real number $\abs{\psi(\x)}^2$ represents ``how much'' of $\psi$ is concentrated at the point $\x$ in the configuration space. When a measurement of the positions of all particles is performed, this ``how much'' becomes\footnote{This ``metamorphosis'' of the squared amplitude into probabilities is one of the foundational problems of Quantum Physics which various interpretations try to solve.} the probability that the wavefunction collapses, as a result of the measurement, at the position $\x$ in the configuration space.
On the other hand, the value of $\Psi$ is a mental state. Therefore, if there is a relation like in eq. \eqref{eq:Psi-and-psi} between the values of $\psi$ and those of $\Psi$, it is evidently a serious open problem how to interpret it.

In Quantum Physics, \emph{nonlocal correlations} are made possible by entanglement, but they become manifest only through quantum measurements, because these correlations are between possible outcomes of the measurements of entangled systems. 
In general we can not directly observe the quantum state. We observe by making quantum measurements, which involve a collapse of the wavefunction, so we can at best learn the collapsed state, not the original one.

On the other hand, for mental states to affect objectively observable physical states, or to extract information about these states, some mechanism is required to endow the mental state with \emph{causal powers}. If eq. \eqref{eq:Psi-and-psi} reflects a true relation between the mental states and the wavefunction, then one may speculate that mental states affect objectively observable physical states by the same wavefunction collapse. Maybe the collapse of the wavefunction, which is postulated to explain both the outcomes of quantum measurements, and the emergence of the classical world, is the way. But, in the absence of an understanding of consciousness, these analogies are wild speculations.

However, even if the relation from eq. \eqref{eq:Psi-and-psi} remains unclear, the analogies between the O-nonlocality of $\Psi$ and the entanglement of $\psi$ (and the nonlocal correlations encoded in it) invite us to experiment in this direction. A first necessary step is to search the possible places in the brain that can maintain coherence, so that, when the wavefunction collapse happens, this can be amplified well enough to result in observable differences in the states of the neurons.

If quantum effects play a role in the brain is still an open question, but they better play a significant one, able to support O-nonlocality, because otherwise we will have to find other ways to support it. And it is difficult to find other ways, because we will also have to make sure that these ways do not allow faster than light or back in time signaling. In Quantum Physics there are theorems that prevent these kinds of signaling, so assuming that the O-nonlocality of mental states is quantum seems a better choice than assuming additional physical laws, and then trying to enforce on them no-signaling.

Regardless of the solution, this by itself does not necessarily solve the Hard Problem of consciousness.

\section{Implications}
\label{s:implications}

The purpose of this article is to discuss the topic of mental states being O-nonlocal. Whatever its implications are, they are not yet developed. So I try though to discuss some possible consequences. 

The first implication pertains quantum nonlocality.
\begin{implication}[quantum in the brain]
\label{implication:quantum}
To avoid postulating new physics, it is reasonable to explore the possibility that Quantum Physics is important for the brain's functioning, and that it may be related to consciousness.
This was discussed in Sec. \sref{s:quantum}. 
If the nature of O-nonlocality is quantum, it should have physical effects, in principle detectable in the brain (see Objection \ref{obj:epiphenomenalism} and the reply).

A possible type of experiment, admittedly not within the reach of present technology, would be to shut down all quantum effects in the brain, \eg entanglement, and see how this affects the mental states. If the classical is all that is needed, one should be able to do this without affecting the mental states.
Possible false negatives can be obtained if epiphenomenalism is true, but I think this position is not interesting, as explained in the reply to Objection \ref{obj:epiphenomenalism}.
I am completely unable to indicate where these quantum effects are to be found in the brain, and what we can do to stop them without affecting the classical approximation of the brain, but I think that by the advance of technology we will be able to learn more about the brain and how to make such experiments.

But maybe some evidence is already present. 
Quantum probability turns out to be more appropriate to model decision making, reasoning, similarity, categorization, and other cognitive processes, which are harder to reconcile with classical (Bayesian) probability, as shown in the extensive review \cite{PothosAndBusemeyer2013CanQuantumProbabilityNewDirectionCognitiveModeling}. According to the authors, ``all these characteristics appear perplexing from a classical perspective. Yet our thesis is that they provide a more accurate and powerful account of certain cognitive processes.``
A way to test for violations of generalized Bell inequalities in human decision making was reported in \cite{CervantesDzhafarov2018SnowQueenIsEvilAndBeautifulExperimentalEvidenceProbabilisticContextualityInHumanChoices}. According to the authors, ``unambiguous experimental evidence for (quantum-like) probabilistic contextuality in psychology'' was found (although in \cite{Atmanspacher2019ContextualityRevisitedSignalingMayDifferFromCommunicating} a loophole is identified).
Follow-up research was presented in \cite{BasievaEtAl2019TrueContextualityBeatsDirectInfluencesInHumanDecisionMaking} and references therein, where contextuality was used. The authors report ``a series of crowdsourcing experiments that exhibit true contextuality in simple decision making''.
All these results show that quantum probability of the type exhibited in entangled systems are more appropriate to model cognitive processes. However, until they are supplemented by a detailed analysis of the mechanisms implementing contextuality and other aspects of quantum probability, which is a very difficult task, I would speculate too much by assuming that its origin is quantum and that they count as evidence for the arguments presented in this article. It seems more plausible that correlations of the same type are due to the shared functionality of the weights of the neural network (although the results are very interesting even if this is the case).

More direct evidence of quantum effects is needed, and promising results are reported in the study of microtubules, as shown in Hameroff and Penrose's review article \cite{HameroffPenrose2017OrchORUpdatedReview}. Their suggestion that anesthesia can shut down consciousness by reducing the quantum effects is strikingly along the lines of the type of experimental results advocated earlier in this Implication to be expected from the arguments in this article.

Despite apparently decisive theoretical arguments that decoherence would prevent any relevant superposition in ``warm, wet and noisy'' systems \cite{tegmark2000decoherenceBrain}, new evidence shows that quantum coherennce can be maintained in such environments, at least for photosynthesis \cite{EngelEtAl2007EvidenceQuantumCoherencePhotosynthesis, PanitchayangkoonEtAl2010LongLivedQuantumCoherencePhotosynthesis}. In the case of the brain, the possibility that nuclear spin quantum processing is possible, by using entangled Posner molecules, was proposed in \cite{MatthewFisher2015QuantumCognitionProcessingWithNuclearSpinsInTheBrain,WeingartenDoraiswamyFisher2016NewSpinOnNeuralProcessingQuantumCognition,JenniferOuellette2016NewSpinOnTheQuantumBrain}.

Nevertheless, even if the quantum will turn out to be necessary for mental states, I think that classical computationalism will still remain very useful in the way Classical Physics is still useful.

While most biological cells, including neurons, do not seem to base their functionality essentially on quantum effects, if this is the case for neurons, it may as well be the case for other types of cells, given that neurons themselves, like the other cells, appeared by specializing from the same original cell.
\end{implication}

Another implication is related to ``objective measures'' of consciousness, defined for physical systems.
\begin{implication}[Integrated Information Theory]
\label{implication:IIT}
In \cite{Tononi2016IntegratedInformationTheory}, it was proposed that consciousness arises from integrated information, which is measured by a function $\Phi$ (not the same as $\Psi$ or $\psi$). This is a measure of how much additional information a system has compared to its subsystems, under certain conditions that are considered relevant to distinguish consciousness from other forms of information. The particular definition was criticized, for example in \cite{Aaronson2014PrettyHardProblemOfConsciousnessWhyIAmNotIIT}.

But integrated information \emph{per se} is an abstract notion, as long as it is not experienced subjectively, ``from within''. It is a property that we, conscious beings, attribute to systems, based on how we assign information to these systems. Theorem \ref{thm:nonlocality-O} shows that mental states are O-nonlocal, so there has to be a (likely) physical way to connect the parts of the system, in order to really integrate them. The implication is that integrated information, to be associated to mental states, requires O-nonlocality.
\end{implication}

Another implication is for the \emph{strong AI} thesis, according to which consciousness is purely computational, and we can create it artificially, on classical computers.
\begin{implication}[no classical strong AI]
\label{implication:strongAI}
It is no doubt that a simulation of the behavior or functionality of mental processes is possible, in principle. What seem to remain outside this possibility are the intrinsic, experiential aspects, although many consider that somehow these should ``emerge'' if we merely simulate the functionality. However, since Claim \ref{claim:mental-classical-computer} is refuted by Corollary \ref{thm:claims_contradiction}, and by the thought experiment from \sref{s:thought-experiment-digital}, consciousness cannot be reduced to classical computation. Hence, it cannot be simulated classically, and the strong AI thesis is 
challenged. And if epiphenomenalism is false (Objection \ref{obj:epiphenomenalism}) and even the simulation of mental processes require quantum effects, it may be possible to refute the Strong AI thesis.
But these results leave open the possibility to simulate mental states on quantum computers.
\end{implication}

\begin{implication}[block-world mental states]
\label{implication:relativity}
Let us go back to the thought experiment from Sec. \sref{s:thought-experiment}. Since mental states are O-nonlocal, then they depend on the observer, because of the \emph{relativity of simultaneity}. Two observers flying in different directions or with different velocities will have different simultaneity spaces, and in each of them the configurations of bits across the cosmic brain are different, so the associated mental states are observer-dependent.
This is a problem.
Indeed, by applying $\Psi$ to the succession of physical states as expressed in different reference frames one expects to obtain distinct successions of mental states. This seems to lead to the strange conclusion that there have to be potentially infinitely many subjective successions of mental states, possibly one for each possible reference frame. Of course, mental states are private, only their physical correlates are objectively verifiable. And if a person reports her own mental states to observers in different reference frames, they will receive the same report, and will agree about it. The private experience is of the person whose brain is observed, and not of the other observers, but the interpretation of the firings of neurons by different observers suggests different successions of mental states. 

One can consider that there is a preferred reference frame, most naturally brain's frame. But the thought experiment in \sref{s:thought-experiment-digital}, in which the space stations are in relative motion with respect to one another, suggests that there is no such frame. However, it is not inconceivable that there is a preferred way to slice spacetime into three-dimensional spaces, which is not objectively manifest.

Another possibility, consistent with Relativity, is that the changes in the mental states are slow enough so that in any reference frame the story will appear essentially the same. This requires that all changes of the mind state are in causal relation, \ie they are ordered so that each change is in the future lightcone of any past change. This total ordering is the situation from Reply \ref{reply:enforced-locality}, which leads to the nonlocality in time from Reply \ref{reply:nonlocality-in-time}.

Or maybe there is a kind of timeless block-world experience, a four-dimensional ``metamental state'', and slicing it in one reference frame or another yields usual mental experiences, related to the physical states by a function $\Psi$ which depends on time and on the reference frame.

Truth is, I do not have an answer to this problem.
The situation is complicated by the fact that the subjective perception of time is very different from the objective measures of time. We have different perception of simultaneity, of durations, and even of ordering of events. The mind seems to revise the perceptual data to construct a representation that makes sense to it, even if it defies the objective notion of time, as explained \eg in \cite{Dennett1993ConsciousnessExplained}, and shown very convincingly in \cite{vanWassenhove2009MindingTimeInAnAmodalRepresentationalSpace}. This by itself may be a challenge to Principle \ref{ppPhysicalMentalCorrespondence}. I am not sure how strong a challenge is, because I distinguish between how we perceive our own experiences, and the fact that we have them, as I extensively explained in the discussion following Observation \ref{ppIntegration}. Because of this distinction, I am also not sure if we can use the data about our illusionary perception and representation of time to decide among the possibilities discussed in this Implication.
\end{implication}

These are some immediate and more relevant implications of the O-nonlocality of mental states. The purpose of this article was only to point out this O-nonlocality, but further research is required to refine our understanding of both the O-nonlocality of mental states and processes, and of its consequences.

\subsection*{Acknowledgement}
The author thanks Igor Salom, Iulian Toader, Johannes Kleiner, Larry Schulman, Cosmin Vișan, Jochen Szangolies, Per Arve, David Chalmers, and the anonymous reviewers, for their valuable feedback, which was not always of agreement, offered to previous versions of the manuscript. The author bares full responsibility for the article.


\appendix

\section{Possible objections}
\label{a:objections}

Here I will discuss some of the possible objections
raised by myself or others against my arguments.
I hope that by this I anticipate the most important objections the reader may have, and address them convincingly.

\subsection{Basic objections}
\label{a:objections-basic}

\begin{objection}
\label{obj:PhysicalMentalCorrespondence}
The problem is that the argument starts from Principle \ref{ppPhysicalMentalCorrespondence}. Why is this assumption justified?
\end{objection}
\begin{reply}
\label{reply:PhysicalMentalCorrespondence}
Principle \ref{ppPhysicalMentalCorrespondence} is simply the reductionist claim that mental states are function of the physical states. It is the minimal assumption about mental states, it does not claim anything about their nature, only that they are function of the physical states. Even if one is not a reductionist, even if one is a dualist, one cannot deny that mental states have physical correlates, so one cannot deny Principle \ref{ppPhysicalMentalCorrespondence} at least for those properties of the mental states that have physical correlates. And by this, according to Theorem \ref{thm:nonlocality-O}, follows that mental states are O-nonlocal, even if one tries to escape physicalism.

Let us verify that the major physicalist positions about the mental processes satisfy Principle \ref{ppPhysicalMentalCorrespondence}.
Blackmore and Tro{\'s}cianko say ``Materialism includes identity theory (which makes mental states identical with brain states) and functionalism
(which equates mental states with functional states)`` \cite{BlackmoreTroscianko2018ConsciousnessAnIntroduction}. 

We start with \emph{functionalism}, the position that only the function and the causal relations matter, and the particular implementation is irrelevant. More precisely, ``functionalist theories take the identity of a mental state to be determined by its causal relations to sensory stimulation, other mental states, and behavior'' \cite{sep-functionalism}. According to Piccinini, ``[t]o a first approximation, functionalism is the metaphysical view that mental states are individuated by their functional relations with mental inputs, outputs, and other mental states'' \cite{Piccinini2004FunctionalismComputationalismAndMentalStates}. Searle characterizes it briefly as ``Mental states are defined by their causal relations'' \cite{Searle1992RediscoveryOfMind}.
According to Goff \cite{Goff2017ConsciousnessAndFundamentalReality}, ``[b]ehaviorists or functionalists believe that the nature of a mental state can be completely captured in causal terms. Causal structuralists generalize this model to the whole of reality, resulting in a kind of \emph{metaphysical behaviorism}. Things are not so much \emph{beings} as \emph{doings}. Pure physicalism is a form of this view.''
If we can represent the functions and causal relations like the \emph{flowchart} of an algorithm or a process, we treat each block and even the algorithm or the process itself as a black box. But any implementation of the functions and the causal relations is done in practice as a process, so Principle \ref{ppPhysicalMentalCorrespondence} holds for functionalism. Moreover, functionalism relies only on classical relations and functions, so it supports Claim \ref{claim:mental-classical}.

\emph{Computationalism} asserts, in addition, that the relations between inputs and outputs are computational \cite{Piccinini2004FunctionalismComputationalismAndMentalStates}. Searle describes the \emph{strong Artificial Intelligence thesis} as ``Mental states are computational states'' \cite{Searle1992RediscoveryOfMind}. Piccinini \cite{Piccinini2004FunctionalismComputationalismAndMentalStates} says that ``[c]omputationalism [...] is precisely the hypothesis that the functional relations between mental inputs, outputs, and internal states are computational. Computationalism per se is neutral on whether those computational relations constitute the nature of mental states.'' Here, for simplicity and because this became the general usage, by \emph{computationalism} I  understand only the ``strong'' version of computationalism, according to which mental states reduce to computation. This strong computationalism can therefore be characterized as the position supporting Claim \ref{claim:mental-classical-computer}.
Since we focus on physicalism, according to Claim \ref{claim:mental2physical-reduction}, functionality is implemented as a physical process, which is a succession of states, and a function like \eqref{eq:psi-mental-states} relates the underlying physical state to the mental state. Hence, both functionalism and computationalism admit a relation between the physical correlate and the corresponding mental state like in eq. \eqref{eq:psi-mental-states}. This should not be a surprise, since from Metaprinciple \ref{ppDynamicalSystem} follows that any physicalist theory of mind has to imply such a relation.

A trending position is that of \emph{externalism} with regard to mental content, according to which ``in order to have certain types of intentional mental states (\eg beliefs), it is necessary to be related to the environment in the right way'' \cite{sep-content-externalism}, and ``[e]xternalism in the philosophy of mind contends that the meaning or content of a thought is partly determined by the environment'' \cite{sep-self-knowledge-externalism}. This position extends the physical state underlying the mental state, but this does not affect Principle \ref{ppPhysicalMentalCorrespondence}. It is correct at least that the interpretation of the relations between the mental representations and the outside objects require, well, these objects, the environment, to be taken into account, as illustrated by Putnam's \emph{Twin Earth thought experiment} \cite{Putnam1975TheMeaningOfMeaning}. This does not imply though that the mental states themselves, as experienced, and not as interpreted in relation with the environment by a third person who tries to give them an intentional or relational meaning in the context of the outside world, have to depend on the environment. But even if it would be true, it is hard to see how externalism can be a physicalist position and at the same time remain local, and some argue that it is not (see \eg \cite{Dretske1997NaturalizingTheMind}, Ch. 5). However, the proof in this paper applies to externalism too, and makes it clear that it has to be nonlocal, endorsing a possible externalist position which requires nonlocality.
Other positions that advocate for nonlocality are known in the literature \cite{LahavShanks1992HowToBeAScientificallyRespectablePropertyDualist,ChrisClarke1995TheNonlocalityOfMind,VanLommel2013NonLocalConsciousness,HaeslerBeauregard2013NDENonlocalMind,HameroffPenrose2017OrchORUpdatedReview,Pylkkanen2018QuantumTheoriesOfConsciousness}, and they can find additional support in the proof I gave here.

In the most general types of physicalism, the function $\Psi$ identifies mental states with coarse-grainings of physical states, \ie with states, possible macro states, that may ignore differences of the microphysical or fundamental states that grounds those mental states. More precisely, two physical states $\tn{\texttt{physical\_state\_1}}$ and $\tn{\texttt{physical\_state\_2}}$ are not distinguished by the function $\Psi$, if $\Psi\(\tn{\texttt{physical\_state\_1}}\)=\Psi\(\tn{\texttt{physical\_state\_2}}\)$.
The position identifying mental states with physical states is a particular case, where $\Psi$ is the identity function.

Other possible objections to Principle \ref{ppPhysicalMentalCorrespondence} can make use of the fact that consciousness integrates visual and auditory stimuli by delaying them to compensate for the different durations of transmissions through the nervous system, and that we are unable to distinguish them instantaneously. But regardless of the particular implementation, Principle \ref{ppPhysicalMentalCorrespondence} refers to the current mental state, even if it emerges with a delay with respect to the stimuli, and even if it does not have a very good resolution of time.
\end{reply}

\begin{objection}
\label{obj:paranormal}
Are you claiming that mental states are paranormal?
\end{objection}
\begin{reply}
\label{reply:paranormal}
There is no paranormal assumption or claim here.
Such misunderstandings are common, and they prompted for example Searle \cite{Searle1992RediscoveryOfMind} to write
\begin{quote}
I have, personally speaking, been accused of holding some crazy doctrine of ``property dualism'' and ``privileged access,'' or believing in ``introspection'' or ``neovitalism'' or even ``mysticism,'' even though I have never, implicitly or explicitly, endorsed any of these views. [...] They think the only real choices available are between some form of materialism and some form of dualism.
\end{quote}

But I do not even reject here the physicalist Claim \ref{claim:mental2physical-reduction}, only the classical physicalist Claims \ref{claim:mental-classical} and \ref{claim:mental-classical-computer}.

Theorem \ref{thm:nonlocality-O} is not a claim I made out of thin air, it is simply the straightforward conclusion of the reductionist Principle \ref{ppPhysicalMentalCorrespondence}. It is not a claim, but a proof that mental states are O-nonlocal, and I explained in what sense and under what condition this O-nonlocality is nonclassical -- that the {Integration Property} holds, which the reader is free to reject. 
Like in the case of quantum nonlocality, the argument itself does not imply that O-nonlocality can be used for faster-than-light signaling or back-in-time signaling. When a mental state is reported by the subject, the resulting message is encoded in a physical state. The physical system arrives in that state through interactions according to the laws of physics.
So nothing ``paranormal'', nothing ``supernatural'' is predicted to be observed.
\end{reply}

\begin{objection}
\label{obj:explain}
Can you explain the role of this O-nonlocality in the mental states? Can you show where in the brain O-nonlocality enters to produce the mental states?
\end{objection}
\begin{reply}
\label{reply:explain}
Theorem \ref{thm:nonlocality-O} proves O-nonlocality, but not how it works to yield the mental states, neither where and how to look for nonlocality in the brain.
This can be the objective of future research. For the moment, it is important to know that O-nonlocality is needed for mental states, but how and why it contributes to mental states remains to be investigated.
\end{reply}

The following Objection comes from a Reviewer.
\begin{objection}[Epiphenomenalism]
\label{obj:epiphenomenalism}
Even if we grant that mental states exist and are integrated (hence nonlocal), they are not necessarily causal. Therefore there is no challenge to physicalism or computationalism, which by definition hold that the physical aspects of thought and behavior are all that matter.
\end{objection}
\begin{reply}
\label{reply:epiphenomenalism}
In my argument, I did not make use of the causal powers of the mental states, only of their determination by the underlying physical states. For this reason, I think these arguments apply even if we assume epiphenomenalism.
This being said, it is indeed logically possible that quantum effects exist in the brain with the only reason to allow the integrated experience of mental states. The functional and computational approaches to mind would not be affected, because the measurable outputs would be the same. The absence of causal powers of the O-nonlocality involved would make it epiphenomenal, in the sense of lacking causal powers. If we take the reporting of the integration of subjective experience as evidence for O-nonlocality, and if the nature of this nonlocality is quantum, it should be able to make empirical predictions. Such effects should be, in principle, detectable by experiments, even if the same functional or computational models can be implemented alternatively based on Classical Physics. A possible type of experiment, definitely too advanced for the current state of technology, can use the elimination of all possibilities of entanglement in the brain in order to see how this affects the mental states.

On a personal note, I think that arguments like the ones in this paper could have led us to the idea of nonlocality even before the discovery of quantum nonlocality. In this case, quantum nonlocality would have appeared to some as an empirical corroboration of Theorem \ref{thm:nonlocality-O}, to be further confirmed by discovering relevant quantum phenomena in the brain. On the other hand, quantum indeterminacy  also appeared to some as a confirmation of the thesis of libertarian free-will, yet there is no consensus in the scientific and philosophical communities that this indeed is the case.
\end{reply}

\begin{objection}
\label{obj:counterexample}
Physicalism, in particular computationalism, define mental states as the states of the underlying physical systems. This gives a counterexample to your arguments.
\end{objection}
\begin{reply}
\label{reply:counterexample}
My argument relies on the assumption that mental states have the {Integration Property} \ref{ppIntegration}. Their conclusion can be avoided by rejecting this assumption, and two alternatives are Options \ref{opt:instrumentalism} and \ref{opt:no-mental-states}, discussed in Appendix \sref{a:options}.
By contrast, people who consider mental states real may reach the opposite conclusion, and see my arguments as a counterexample to the theories from Objection \ref{obj:counterexample}.
\end{reply}

\subsection{Objections related to extension in space}
\label{a:objections-space}

\begin{objection}
\label{obj:global}
Why should O-nonlocality necessarily be nonclassical? After all, there are nonlocal properties, in the sense that they depend on spacelike separated subsystems, even in local theories. For example, the total charge, mass, energy, and momentum of a system depend on all the constituents of the system, or on the values of its fields at all points it occupies in space. Similarly, the center of mass. 
Temperature is proportional to the average molecular kinetic energy over a large number of particles constituting the gas or liquid or solid.
\end{objection}

Before replying to Objection \ref{obj:global}, let us define ``classical nonlocality'',

\begin{definition}[C-nonlocality]
\label{def:nonlocality-C}
A property of a system is called \emph{C-nonlocal} if it depends on spacelike separated subsystems.
\end{definition}

\begin{remark}
\label{rem:nonlocality-C}
We see from Definitions \ref{def:nonlocality-C} and \ref{def:nonlocality-O} that an O-nonlocal property is a directly observable C-property.
\end{remark}

\begin{reply}[to Objection \ref{obj:global}]
\label{reply:global}
Local theories allow C-nonlocal properties, but such theories are still local, because \emph{their processes can be fully described in terms of the local properties, without ever mentioning the C-nonlocal ones}.

From the examples given in Objection \ref{obj:global}, let us discuss temperature, which is observable, though not directly observable. It is hard to see how temperature, which is successfully described in a reductionist way in terms of Newtonian Mechanics, is a sort of illusion. It is real, and since the average is taken over a large volume, it is C-nonlocal. But in fact its effects, in particular on the skin of a human being who may feel burnt, or on the mercury rising in a thermometer, are completely local effects. It is only in our mind that these effects are integrated and unified, \cf Observation \ref{ppIntegration}. The skin is really burnt, but this effect is local. The mercury indicating the temperature rises in the thermometer, but everything happens through local interactions, including the perception of the level of the mercury we may have. The only integration of this information into the idea of temperature takes place in our minds. Mind only has indirect access to quantities like total charge or temperature. But mind experiences directly its own state, and this is where the difference becomes relevant. Therefore, ultimately, according to Definition \ref{def:observer}, even in these examples, mental states are the O-nonlocal ones.

It is hard to say that for the total mass there is an experience of awareness of that mass, unless that mass is measured somehow and the result communicated to a sentient observer. But in the case of a mental state grounded on region $R$, because of Observation \ref{ppIntegration}, there is instantaneous integrated experience determined by the physical state from $R$. This experience has the {Integration Property}, and it is instantaneous, according to Principle \ref{ppPhysicalMentalCorrespondence}.
And we know directly our own mental states, we have the direct integrated experience of them, they are this very experience. Even if we fall for folk psychology, even if our mental state is not exactly how we imagine it to be, even if we are silly enough to imagine it as a homunculus, what is for sure is that the experience is there, and this is what I mean by mental state. And this depends on the spacelike separated regions $A$ and $B$.
It is not scattered in a separate, independent way, in regions $A$ and $B$. If we would have to admit such a scattering, we would have in fact to admit that the experience is divided in all points of the region $R$, and this would mean no experience at all.
Because we know that experience is present, we conclude that the mental state is O-nonlocal, unlike the classical C-nonlocal properties. 
\end{reply}

\begin{objection}
\label{obj:extended_objects}
Consider a chair in Classical Mechanics, or any other rigid object extended in space, or even a soft one like a brain. This does not make it O-nonlocal, so why would mental states be O-nonlocal?
\end{objection}
\begin{reply}
\label{reply:extended_objects}
An extended object like a chair is made of particles that interact locally, and these interactions maintain the chair's configuration stable enough in time. There is nothing O-nonlocal here, except our notion of rigid object like a chair. Being extended in space is a C-nonlocal property, and the laws of Classical Physics are local. If we move the chair, we apply forces that propagate locally within the chair, and make all of its parts move in an apparently rigid way. If we sit on the chair, it supports our weight again by local interaction between its atoms. No O-nonlocality is required for it to exist as an approximately rigid object. The chair exists as a unified simultaneous configuration only in our minds, \cf {Observation} \ref{ppIntegration}, and this mental representation in our mental states is O-nonlocal, but not the chair.
\end{reply}

\begin{objection}
\label{obj:classical_computer}
A classical computer is extended in space, and is able to process information, independently of our mental representations, contrary to the claims made in Reply \ref{reply:extended_objects}.
\end{objection}
\begin{reply}
\label{reply:classical_computer}
In fact, this is a good example for Reply \ref{reply:extended_objects}. A classical computer is just a system that evolves in time according to the laws of physics. Its discrete states are physical states. They are seen as carrying information only by us, the users of the computer. We are the ones who interpret the inputs and outputs as information, but for the computer as a physical system they are merely initial and final states. We are the ones who interpret its time evolution as information processing. This computation seems objective, and is objective, in the sense that different users will agree that the same information processing takes place in the same computer. But this kind of objectivity arises from the fact that the underlying physical processes are objectively the same for all of us, and from the fact that we share the same conventions about the information that goes in and out of the computer.
Ultimately the attribution of experience to such a scattered, disconnected physical state is an illusion, a feature imagined in the very mental state entertaining this possibility, perhaps correlated to the activity of the mirror neurons.
\end{reply}

We see from the Replies to Objections \ref{obj:global}, \ref{obj:extended_objects}, and \ref{obj:classical_computer}, and from Remark \ref{rem:nonlocality-C}, that the difference between O-nonlocality and C-nonlocality is due to the {Integration Property}.

\subsection{Can locality be enforced?}
\label{a:enforced-locality}

In this subsection I discuss some objections suggesting ways to prevent O-nonlocality by enforcing locality.

\begin{objection}
\label{obj:gates}
I think that an actual computer, especially if it is supposed to simulate human-like brains, is more complicated, and cannot be spread around the galaxy the way you claim in \sref{s:thought-experiment-digital}, a small number of bits on a separate space station gravitating a different star.
\end{objection}
\begin{reply}
\label{reply:gates}
It is known that the \textbf{NAND} gate is \emph{functionally complete}, \ie any logical circuit can be made of such gates.
You can realize a fully functional computer using only \textbf{NAND} gates, like the one built by Kevin Horton \cite{kevtris2013NANDputer}.
The \textbf{NAND} logical operation takes as input two bits $\tn{p}$ and $\tn{q}$, and outputs the negation of their conjunction, $\overline{\tn{p} \cdot \tn{q}}$ (see Fig. \ref{fig:nand-gate}).

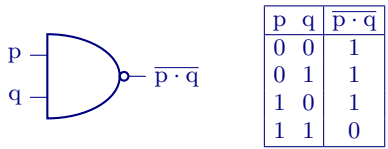
\begin{figure}[!ht]
\centering
\begin{circuitikz}
    \matrix[column sep=2em, ampersand replacement=\&,color=darkblue] {
				\draw (0,0) node[nand port] (mynand) {}
				(mynand.in 1) node[anchor=east] {p}
				(mynand.in 2) node[anchor=east] {q}
				(mynand.out) node[anchor=west] {$\overline{\tn{p} \cdot \tn{q}}$};
        \&
        \node {
            $\begin{array}{|cc|c|}
						\hline
                \tn{p} & \tn{q} & \overline{\tn{p} \cdot \tn{q}} \\ \hline
                0 & 0 & 1 \\
                0 & 1 & 1 \\
                1 & 0 & 1 \\
                1 & 1 & 0 \\
								\hline
            \end{array}$
            };
        \\
    };
\end{circuitikz}
\caption{\small{\label{fig:nand-gate}{The \textbf{NAND} gate, which is an inverted \textbf{AND} gate.}}}
\end{figure}

Let us place a single \textbf{NAND} gate on a different space station, orbiting a separate star. 
One may object that it is difficult to synchronize the timings when the bits $\tn{p}$ and $\tn{q}$ are received by the \textbf{NAND} gate, given the large distances between the stars. This is possible, and I will explain how it can be done y allowing each space station to store two bits, $b_1$ and $b_2$. Suppose the bits $\tn{p}$ and $\tn{q}$ come in sequence. It does not matter which of them comes first, because the operation $\overline{\tn{p} \cdot \tn{q}}$ is commutative. We use the bit $b_1$ to store the information that only one of the two bits $\tn{p}$ and $\tn{q}$ is collected. So when $b_1=0$, it means that no bit was collected. As soon as one of the two bits $\tn{p}$ and $\tn{q}$, say bit $\tn{p}$ arrives, it is copied in $b_2$, and we make $b_1=1$. When the bit $\tn{q}$ arrives, if $b_1=1$, we send both bits $\tn{p}$ and $\tn{q}$ through the \textbf{NAND} gate, and send the resulting output, by using an electromagnetic signal, to the next star in the circuit. Then we reset $b_1$ to $0$, indicating that the setup is ready for the next logical operation. If $b_1=0$ when the bit $\tn{q}$ arrives, it means that the bit $\tn{p}$ did not arrive yet, so we copy $\tn{q}$ in $b_2$ and make $b_1=1$, then wait for $\tn{p}$.
Therefore, two-bit components are sufficient. But we can go even deeper, considering that each logical gate is made of several diodes, resistors, and transistors\footnote{Diodes and transistors use quantum effects, but the logic gates can, in principle be implemented using Classical Mechanics.}. They usually are part of integrated circuits, but the computation is the same even if we place these electronic components on separate space stations around different stars.

We may prefer, of course, to use more bits to take care of the exchanged signals, to send them, maybe to receive confirmation that they arrived, maybe to send them repeatedly to implement error correction \etc, but it can be assumed that this is not part of the computation itself. But regardless of these details, the point is that it can be arranged so that no space station does a computation sufficiently complex to allow us to attribute a human-like mental state to that station alone. And let us not ignore the fact that even if the two bits $b_1$ and $b_2$ are very close to each other on the space station, the simultaneous events that they can support are spacelike separated.
\end{reply}

\begin{objection}
\label{obj:enforced-larger-gates}
Maybe a computer like in Reply \ref{reply:gates} cannot support local mental states. But such a computer was considered with the purpose to be maximally spreadable in space. What if we consider instead one that has components that are not so spreadable? Wouldn't such a computer be immune to the argument from \sref{s:thought-experiment-digital}?
\end{objection}
\begin{reply}
\label{reply:enforced-larger-gates}
This would not help, and here is why. The point of spreading the parts of the computer across the galaxy was not to ``destroy'' its mental states, but to show that they require O-nonlocality. It is not necessary to build it out of \textbf{NAND} gates to show this. 
First, any computer can be realized like this. Any computer is \emph{Turing equivalent} to one made solely out of \textbf{NAND} gates.
Second, its components are already spread in space. One can try to miniaturize the computer how much you want, its parts will still be spacelike separated. Spacelike separation does not depend on the scale. No matter how close they are in space, two simultaneous events are spacelike separated. The thought experiment from \sref{s:thought-experiment-digital} just emphasizes this separation, but it does not introduce it.
So no matter how complicated logical gates one will use, they will always execute bit operations separated in space.
\end{reply}

\begin{objection}
\label{obj:enforced-locality}
Since Theorem \ref{thm:nonlocality-O} assumes Observation \ref{ppMentalExtended}, before accepting O-nonlocality, I think we should try more and see if it is possible to avoid this assumption, by constructing a centralized, and therefore local, model of mental states.
As long as we did not exclude without a shred of doubt this possibility, the conclusion of Theorem \ref{thm:nonlocality-O} is not justified.
\end{objection}
\begin{reply}
\label{reply:enforced-locality}
This is a good point, and I will do my best to make the case for enforced locality, before replying.

Let us start by using a classical computer architecture as a very rudimentary model of the brain, according to Claim \ref{claim:mental-classical-computer}.
I will use for reference the \emph{Harvard architecture} (see Fig. \ref{cpu-harvard-architecture.png}), rather than the \emph{von Neumann architecture}, because it appears to be more centralized.

\image{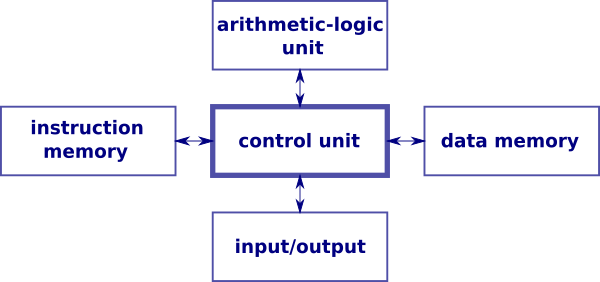}{0.45}{Harvard computer architecture.}

Can we take the existence of a centralizing unit, the \emph{control unit}, as a model of true centralization? \emph{Harvard's Mark I}, the first computer to use this architecture, weighed $4 300$ kg. Its sequence-control unit read a $24$-bit wide \emph{punched paper tape} and followed the instructions to operate the shafts. So it seems that even this one required more than one bit, but, given that a Turing machine can operate with one bit wide tapes, maybe it did not require, maybe it was just for practical reasons.

Is it possible to have something like a central or control unit, or whatever unit one may assume to be associated to the mental states, that requires only one bit of memory?
Given that \emph{a serial computer operates only one bit for each clock cycle}, we can take this bit as underlying the mental state. This is an extreme claim, but let us pursue this line of reasoning anyway.
A one-bit unit supporting mental states would eliminate the option of distributed, decentralized storage of that unit by using more stars. 
And we seem to be lucky with this one, because recently researchers achieved presumably the most local way to store a bit, by using a single atom \cite{Natterer2017ReadingWritingSingleAtomMagnets}. They were able to store a bit in the state of a Holmium atom, invert that bit, and even store different bits independently in different Holmium atoms placed at one nanometer apart, which makes this kind of storage scalable.

In this way, we can try to build a model in which the physical correlate underlying the mental state at a given time is that particular single bit that is operated at that particular instant of time. Whatever a serial computer does, it does it as a sequence of one-bit binary operations. 

The brain is different from a computer, in at least two major ways. First, the brain does not seem to operate logical gates, it is a \emph{neural network}, one that can change itself. But such neural networks can be simulated by digital computers to any degree of approximation.
The other difference is that the brain operates in parallel. But parallel computing exists too, and it still requires a central unit that breaks the task into smaller tasks, and then centralizes the results. And that unit operates sequentially, one-bit-at-a-time. Hence, we can still assume that a one-bit-at-a-time model of the brain is not excluded, at least if we assume independence of the substrate. 

We can therefore imagine that something like one-bit-at-a-time processing can happen in the brain, as illustrated very schematically in Fig. \ref{mental-states-correlates-centralized.png}. We do not have evidence that this is the case, but it worth entertaining this hypothesis, if the gain is the avoidance of O-nonlocal mental states from Theorem \ref{thm:nonlocality-O}.

\image{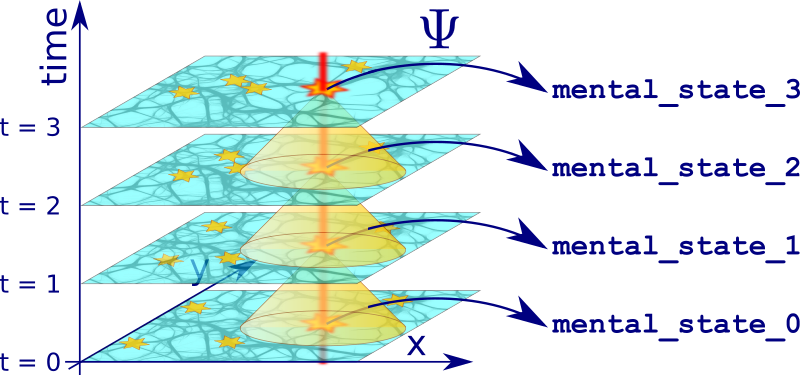}{0.45}{An attempt to make mental states local.}

Let us now assume that each bit operated in this sequence is encoded into a precisely localized physical system. True, Quantum Physics prevents it to be localized at a precise position. Apparently, Heisenberg's \emph{uncertainty principle} allows it to have a precise position, but the price for this is a completely uncertain momentum. This uncertainty of the momentum will make the immediately next position completely uncertain too. But maybe our system can be made to collapse periodically in a precise place. If we go in this direction, \emph{we invoke quantum effects for merely storing a classical bit!} And even if our computer is classical, invoking quantum mechanical collapse in order to obtained localized consciousness
\begin{enumerate}
	\item 
	violates Claim \ref{claim:mental-classical}, by appealing to quantum effects to allow the existence of mental states,
	\item 
	appeals to quantum nonlocality, in an attempt to avoid  O-nonlocality.
\end{enumerate}

So our bit has to be stored in a small region, but not a precisely localized one. A Holmium atom in a relatively stable state should do it. And we can be happy with the situation, because such an atom is approximately local. But it is not quite local, while the wavefunctions of its particles are concentrated around the nucleus, they in fact extend in space to infinity! Flipping a bit encoded in an atom involves changes of the wavefunction that extend to infinite distances. Therefore, this ``approximate locality'' is still nonlocal. Even if in practice there is no way to probe this nonlocality, because our instruments are also quantum, and the wavefunction collapses, and Heisenberg's uncertainty principle kicks in.

But we can try to do even more, and assume an interpretation of Quantum Physics which allows point-particles -- the Pilot-Wave Theory. This way, the bit can be stored in the state of one of the well-localized point-electrons of the Holmium atom. But the Pilot-Wave Theory is a paradigmatic example of nonlocal theory \cite{Bell2004SpeakableUnspeakable}: \emph{the point-particles are guided by the pilot-wave, in a way that depends on the positions of other point-particles in a nonlocal way!}

In addition, there is a severe limit of the amount of information one can store and extract from a region of space, the \emph{Bekenstein bound} \cite{Bekenstein1981UniversalUpperBound,Bekenstein2005HowDoesEntropyInformationBoundWork,Casini2008RelativeEntropyAndTheBekensteinBound,Bousso2002HolographicPrinciple}.
And it turns out that the information that can be stored at a point of space is basically zero. No matter how efficient we do it, suppose we can reduce the implementation of the bits to the Plank scale, we will always need an extended region of space to store even a single bit.

So it appears that \emph{we cannot get rid of nonlocality, even if we assume that mental states are determined by a single bit!}
\end{reply}

It is true, this type of nonlocality may seem weaker,
and it is weaker, it doesn't even contradict Claims \ref{claim:mental-classical} and \ref{claim:mental-classical-computer}. But
the price is that we are limited to one-bit mental states only!
And what kind of mental state, what kind of experience requires only one bit? Is our experience so simple, that a single bit is enough to support or describe it? What solution would be more crazy, that mental states are O-nonlocal, or that they can be stored in an electron (in a still nonlocal way)? Subjective empirical observations show the following (confirmed by Damasio's quote given in support of Observation \ref{ppMentalExtended}):
\begin{observation}[diversity of mental states]
\label{ppDiversity}
Our mental states are complex and diverse, and there are definitely more than two possible mental states.
\end{observation}

Note that even if we take the position that the substrate, the material of the brain, is essential to support mental states, we still have to localize it at the level of a single atom or a particle, and Quantum Physics still does not allow true locality, just like I explained above, and this kind of ``localized homunculus'' does not seem plausible.

\subsection{Objections related to time}
\label{a:objections-time}

It remains a possible way out of the problem of diversity of mental states posed by Observation \ref{ppDiversity}: assume that mental states are as diverse as needed, but they are not determined by a single bit, or by the state of a single particle or even an atom, but \emph{are determined by a sequence in time of the values that bit can have}.

\begin{objection}
\label{obj:nonlocality-in-time}
Even if assuming that mental states are instantaneous implies that they are O-nonlocal, maybe they are not instantaneous, maybe they are in fact extended in time. This violates both Principle \ref{ppPhysicalMentalCorrespondence} and Observation \ref{ppInstantism}, and therefore prevents you from reaching the conclusion of Theorem \ref{thm:nonlocality-O}.
\end{objection}
\begin{reply}
\label{reply:nonlocality-in-time}
The claim that mental states are not instantaneous contradicts Metaprinciple \ref{ppDynamicalSystem}, according to which the physical
world is described by a dynamical system. It is hard to see how denying this Metaprinciple can allow someone to still hold a materialist, or even a physicalist position. But let us go this way, for the sake of the argument, and suppose that mental states can only exist extended in time, as in Fig. \ref{mental-states-correlates-extended-in-time.png}. 

\image{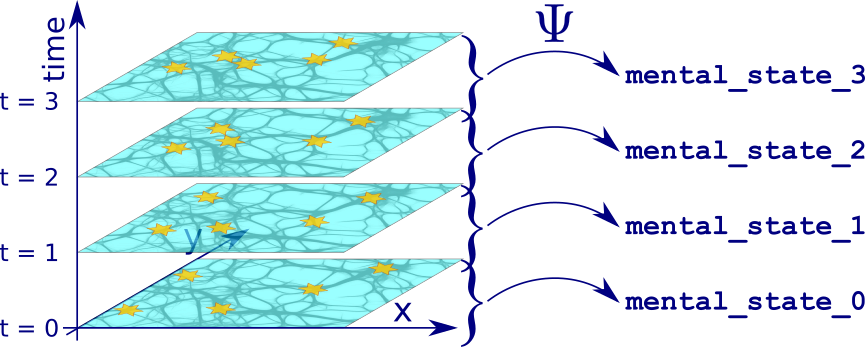}{0.45}{If mental states are extended in time, then they are nonlocal in time.}

Then, Eq. \eqref{eq:psi-mental-states} should be replaced with an equation that expresses a dependency of mental states of the physical states at different times, for example in an interval $[t_i,t_{i+1}]\subset \R$,
\begin{equation}
\label{eq:psi-mental-states-time-extended}
\tn{\texttt{mental\_state}}(t_i)=\Psi\(\tn{\texttt{physical\_state}}|_{[t_i,t_{i+1}]}\)
\end{equation}

But this would only make mental states be nonlocal in time. 
To better understand this, it may help to apply the thought experiment from Sec. \sref{s:thought-experiment} to this argument of temporal extension of mental states.
\end{reply}

\begin{objection}
\label{obj:philosophy-of-time}
I cannot accept your argument because my philosophical position about time is \emph{eternalism} (or by contrary, \emph{presentism}), which is in contradiction with Observation \ref{ppInstantism}.
\end{objection}
\begin{reply}
\label{reply:philosophy-of-time}
Instantism is neutral to the dispute between \emph{presentism} (the position that only the present time exists, and the world changes in time) \vs \emph{eternalism} (the position that all instants of time are equally real and immutable, but each instance of ourselves experiences its own instant as the present).
Instantism is consistent with both of these positions, being simply a direct consequence of the possibility to express the physical laws in terms of dynamical systems, which is consistent with both presentism and eternalism.

Moreover, a common argument raised by some presentists against eternalism is that there can be no experience of time in the block universe of eternalism. But instantism shows that whatever explains our experience of time, the explanation has to apply to each instant, so the very reason that is assumed to allow the experience of time in presentism has to remain the same in eternalism.
\end{reply}

\begin{objection}
\label{obj:records}
If Observation \ref{ppInstantism} is true and we can only access our present state, then how is it even possible to remember the past? How is it even possible to do Science?
\end{objection}
\begin{reply}
\label{reply:records}
Observation \ref{ppInstantism} is trivially correct, denying it is out of discussion, but Objection \ref{obj:records} is nevertheless important.
The problem that the present state can access directly only itself is well-known and not specific to my arguments, but a general problem of physics. The explanation is related to the \emph{arrow of time}. The hypothesis that the universe was in a very special, low-entropy state (presumably at the Big Bang), sometimes called the \emph{Past Hypothesis}, explains the overall increase of entropy in one direction of time (which by definition is called \emph{future}), and other arrows of time.

This hypothesis is required by the understanding of Thermodynamics in terms of Statistical Mechanics, and solves several problems. In particular, it is thought to explain that we know much more about the past compared to the future because the present contains records of past events. These records appear as patterns in the macro states, and they would be indistinguishable or ambiguous if the Past Hypothesis would not be true. This is so because we would not know to interpret them as memories of the past events; they could as well indicate future events or be mere statistical fluctuations. 

There is plenty of literature addressing this problem of time asymmetry and records convincingly, here is a selection \cite{Eddington1928NatureOfThePhysicalWorld,Feynman1965TheCharacterOfPhysicalLaw,Penrose1989EmperorsNewMind,DavidZAlbert2003TimeAndChance,SeanCarroll2010FromEternityToHere,DavidZAlbert2015AfterPhysics,BarryLoewer2016MentaculusVision}. However, these problems are not yet completely understood. But Observation \ref{ppInstantism} is still correct, and any explanations of the arrows of time have to take this into account.
\end{reply}

\section{Alternative options}
\label{a:options}

In this section I discuss possible alternative options, including some that seem to escape the argument that mental states are O-nonlocal.

First, let us identify the possible alternatives to O-nonlocality by looking into the assumptions of Theorem \ref{thm:nonlocality-O} and of Corollary \ref{thm:claims_contradiction}. Corollary \ref{thm:claims_contradiction} establishes that each of Claims \ref{claim:mental-classical} is \ref{claim:mental-classical-computer} are inconsistent with the {Integration Property}. If we want to save classicality, we can reject Principle \ref{ppPhysicalMentalCorrespondence}, the {Integration Property}, or Claim \ref{claim:mental2physical-reduction}. 

\begin{option}[deny Principle \ref{ppPhysicalMentalCorrespondence} or the {Integration Property}, instrumentalist version]
\label{opt:instrumentalism}
The role of Science is only to give an \emph{instrumentalist reduction} of mental states to physical states, but not necessarily an \emph{ontological reduction}.
Consequently, as long as the claimed O-nonlocality of mental states does not lead to faster-than-light or back-in-time signaling or other violations of the known physical laws, we should not even care about the mental states, because they are not objective properties.
Since private experience is not objectively verifiable, it does not need to be explained.
One should only care about the physical states, which are observable, and our language should not be contaminated by statements about mind, consciousness, \etc.
\end{option}
\begin{comment}
\label{comm:instrumentalism}
The role and definition of Science may be under debate. But if we assume that there is a reality, and that our knowledge is entitled to seek explanations about that reality, then a mere definition of Science as Instrumentalism cannot make the questions go away. The arguments in this article are directed towards the understanding of reality, which includes the mental states. In fact, mental states are the way we even know about reality.
If we are committed to Instrumentalism, we have to acknowledge that the primary instrument used in exploring objectivity is subjective:
\begin{definition}[Objectivity]
\label{def:objectivity}
When a property of a system is in principle accessible to all observers, and those observing it agree on its value, the property and its observation are said to be \emph{objective},
even if it may be just an approximation of reality.
Properties inaccessible even in principle to more than a single observer are called \emph{subjective}.
\end{definition}

Science can include among its goals also the explanation of mental states, and conversely, mind is required to do Science.
This is not a singular case in Science, where we often obtain relations that are not one-directional, for example in General Relativity the distribution of matter affects the spacetime curvature, which in turn affects how matter is distributed. We learn about the world through our minds, and about our minds by researching our world. We hope we can break out of this loop into an ascending spiral, and this turned out to be the case so far. But our access to this spiral is not through some \emph{God's eye view} that provides us full objectivity, as we sometimes tend to see Science, but rather trough the \emph{frog's eye view} of our limited and biased mental states. Therefore, I think there is justification also for the epistemically modest view that what we interpret as \emph{objectivity} is merely the consensus emerging from corroborating more subjective views, as captured in Definition \ref{def:objectivity}. 

Whether according to human-made definitions, the questions about mental states lie in the field of Science or of Philosophy, these are questions about things we experience as real, and it is justified to ask them and to try to answer them,
even if they eventually will turn out to ``only seem to be real'' as Dennett phrased it \cite{Dennett1993ConsciousnessExplained}. 

A position like Dennett's \emph{heterophenomenology}, or any similar position which demarcates Science as being only about the objectively observable, is perfectly justified as long as it acknowledges this demarcation. This assumption is very useful in allowing Science to focus on what we can prove to others, but there is no objective proof that this way can exhaust all aspects of reality. 
\end{comment}

\begin{option}[deny Principle \ref{ppPhysicalMentalCorrespondence} or the {Integration Property}, illusionist version]
\label{opt:no-mental-states}
There are simply no mental states, period. This will do away with the whole problem, because if there are no mental states, there is no need for the function $\Psi$ postulated in Principle \ref{ppPhysicalMentalCorrespondence}, such a function would make no sense at all. Mental states are an illusion of some physical states arranged in the right configuration that corresponds to that illusion.
\end{option}
\begin{comment}
\label{comm:no-mental-states}
I confess that I have no objective proof that Option \ref{opt:no-mental-states} is not available. 
But I also confess that I have no idea what even means that there are no mental states, or that they are illusions. Mental states may include wrong or unreal impressions, illusions about who you are, about your mental states themselves, but even if these are illusions, there is an experience of these illusions.
Bill Nye said it very concisely, with his characteristic wit \cite{BillNye2013ThinkingThatImThinking}
\begin{quote}
Am I thinking, or am I just thinking that I'm thinking?
\end{quote}

Taking Dennett's arguments as an example \cite{Dennett1993ConsciousnessExplained}, I fully agree that they show that there is no homunculus, no Cartesian Theater, and even no objectively observable {Integration Property}. His arguments can be interpreted the way Dennett and many of his readers do, as showing that mental states are not real, they ``only seem to be real''. But the same arguments can be interpreted as showing that the classical models are inherently incomplete, precisely for lacking the {Integration Property} and for deeming as unreal the primary tools we use to do Science, our own minds.
Conversely, the very arguments in this article can be interpreted by someone who denies the reality of mental states as providing even more direct and more general evidence for their non-existence. Unless of course my argument makes empirical predictions, as I argue to be expected in Implication \ref{implication:quantum}.

At any rate, no matter how much I would try to convince someone that I do have mental states, that person can always appeal to objectivity and say that the only things that exist are the measurable things, the measurable properties of physical systems. That person (a physical system for that matter) is free to deny the very existence of mental states, because they are not directly measurable. While I can report my mental states to other people, the very process of reporting them means turning them into physical means of communication (like sounds or written text), resulting in physical processes, and it can always be said that this is all there is.
\end{comment}

Sometimes, it becomes very difficult to distinguish among Options \ref{opt:instrumentalism} and \ref{opt:no-mental-states}.
As Searle put it \cite{Searle1992RediscoveryOfMind},
\begin{quote}
Very few people are willing to come right out and say that consciousness does not exist. But it has recently become common for authors to redefine the notion of consciousness so that it no longer refers to actual conscious states, that is, inner, subjective, qualitative, first-person mental states, but rather to publicly observable third-person phenomena. Such authors pretend to think that consciousness exists, but in fact they end up denying its existence.
\end{quote}

Theorem \ref{thm:nonlocality-O} does not assume directly Claim \ref{claim:mental2physical-reduction}, it assumes Principle \ref{ppPhysicalMentalCorrespondence}, which, while being a consequence of Claim \ref{claim:mental2physical-reduction}, it can be true in any approach to the mind-body problem which admits a relation between the mental states and the physical states. So in fact Corollary \ref{thm:claims_contradiction} can be generalized to
\begin{corollary}
\label{thm:locality-vs-mind-body-relation}
Given {Observation} \ref{ppIntegration}, locality and the existence of physical correlates of mental states cannot both be true.
\end{corollary}
\begin{proof}
Principle \ref{ppPhysicalMentalCorrespondence} states that there is a correspondence $\Psi$ between physical and mental states, see eq. \eqref{eq:psi-mental-states}.
But $\Psi$ is not necessarily a function of the form $f:A\to B$, since it is probably not defined for all physical states, because not all physical states appear to support mental states. It is a partial function, because it is defined on a subset of all possible physical states. Moreover, it is as well hard to prove that $\Psi$ is surjective, there may be mental states without physical correspondent. But in all these cases, $\Psi$ is a relation, \ie a subset of all the pairs of the form $\(\tn{\texttt{mental\_state}},\tn{\texttt{physical\_state}}\)$, of the Cartesian product between the set of all physical states and the set of all mental states. Even in this general case, Principle \ref{ppPhysicalMentalCorrespondence} can be true, unless the relation $\Psi$ is an empty subset of this Cartesian product, or a small subset that makes the relation between physical and mental states irrelevant. But as long as there is a relevant relation, \ie as long as the physical state has something to say about the mental state, even if it does not completely determine it, Principle \ref{ppPhysicalMentalCorrespondence} is true and Theorem \ref{thm:nonlocality-O} can be applied, with the consequence that O-nonlocality is required. So to enforce the locality of the mental states, we have to make them supported by one-bit or one-particle physical states, which means to make the relation $\Psi$ a small subset of the Cartesian product, \ie to reject the existence of such a relation between physical and mental states, contrary to Observation \ref{ppDiversity} of \nameref*{ppDiversity}.
\end{proof}

This suggests other options to avoid O-nonlocality.

\begin{option}[deny Claim \ref{claim:mental2physical-reduction}]
\label{opt:nonphysical_nonlocality}
Assume \emph{dualism}, \emph{property dualism}, \emph{panpsychism}, \emph{neutral monism} or other vews about consciousness that deny physicalism.
From Corollary \ref{thm:locality-vs-mind-body-relation}, we see that Claim \ref{claim:mental2physical-reduction} is not necessary to reach the conclusion that mental states are O-nonlocal. What needs to be assumed for this to hold are Observations \ref{ppMentalExtended} and \ref{ppIntegration}, and a weak form of Principle \ref{ppPhysicalMentalCorrespondence}, in which mental states are at least partially a function of the physical states, or in which there is a relation, in the mathematical sense described in the proof of the Corollary. This relation does not have to involve the full mental state. 
The consequence of the generalization of Corollary \ref{thm:claims_contradiction} to Corollary \ref{thm:locality-vs-mind-body-relation} is that mental states have to be O-nonlocal even in other theories about consciousness like those mentioned in Option \ref{opt:nonphysical_nonlocality}. Now, if in these theories there are ``nonphysical mental properties'', they can be the ones to support entirely the O-nonlocality, and the ``physical properties'' can remain classical. So these theories do not require quantum effects, since they already made other assumptions, beyond both classical and quantum physicalism.
\end{option}
\begin{comment}
\label{comm:nonphysical_nonlocality}
This is a way to avoid nonlocal physics, but without saving even the more general Claim \ref{claim:mental2physical-reduction}, because dualism introduces nonphysical properties. I cannot even imagine what would be nonphysical properties, since as long as they are consistently describable by propositions, we can include them among the physical ones, by extending what we mean by physics.
On the other hand, O-nonlocality is allowed by Quantum Physics, the price being to give up Claim \ref{claim:mental-classical}.
I do not see how this alternative option would work, or how it would be able to do more than quantum physicalism, unless it has other advantages, the most desirable one being to explain sentient experience. But for the moment no explanation of sentient experience is known \cite{Sto2020NegativeWayToSentience}.
\end{comment}
%
%

\begin{option}[deny Principle \ref{ppPhysicalMentalCorrespondence} and Claim \ref{claim:mental2physical-reduction}]
\label{opt:bound_state}
This Option is similar to Option \ref{opt:nonphysical_nonlocality}, but, in addition, it assumes that mental state can act on the physical world through very localized inputs in the brain, in a local way consistent with the arguments from \sref{a:enforced-locality}. It does not violate Observation \ref{ppDiversity}, by being richer than what the atom or particle or bit through which the brain accepts them as inputs can support. Imagine a computer whose input is a one-button keyboard, and you can input using the Morse code. One can input any text like this. So, if the $\tn{\texttt{physical\_state}}$ from eq. \eqref{eq:psi-mental-states} is just the input to the brain, not the true physical correlate of the mental state, rich mental states can exist and express themselves physically.
\end{option}
\begin{comment}
\label{comm:bound_state}
Please refer to Comment \ref{comm:nonphysical_nonlocality}.
While this kind of dualism is in principle possible, I personally think it is premature to accept it. Other options allow us to know more about consciousness, and their exploration is not yet exhausted. Accepting this kind of dualism would mean to give up too early. While I cannot refute this option, I find it unhelpful, at least for the moment. 
\end{comment}

\begin{option}[single-bit temporal sequence]
\label{opt:nonlocality-in-time}
The physical system underlying mental states contains only one bit at a time, but the value of the bit changes in time, and the sequence of these bits is what underlies the mental states.
\end{option}
\begin{comment}
\label{comm:nonlocality-in-time}
This option was discussed in \sref{a:enforced-locality}, and in Objection \ref{obj:nonlocality-in-time}, with the conclusion that this would still be nonlocality, albeit in time.
\end{comment}

\begin{option}
\label{opt:holomorphic}
There is a way to store and access as much information as needed at a single point of space: \emph{holomorphic functions}. An analytic or holomorphic function is defined on an extended space, but its values, and the values of its partial derivatives at each point of space, can be used to determine the value of the function at any other point, by using power series expansions at that point.
\end{option}
\begin{comment}
\label{comm:holomorphic}
This is mathematically true. In fact, Stoica \cite{Stoica2017IndrasNetHolomorphicFundamentalness} argues that the fundamental laws of physics may be holomorphic, and that the ontology is not distributed in space or spacetime, but it is all concentrated in a \emph{germ} of the holomorphic field, or an equivalence class of germs of the holomorphic field, from which the fields can be recovered by power series expansion and analytic continuation.

A problem with this option is that there is no known way to access even the precise value, let alone the values of the higher order partial derivatives of the hypothetical field, and if this would be possible, it has to work only up to a certain, unknown limit, because otherwise it could be used to violate faster-than-light or back-in-time no-signaling. There is no known mechanism other than quantum measurement to extract even partial information from the germ. But, for this discussion, the most important aspect is that the germ of a holomorphic field is essentially nonlocal anyway. So, even if this Option would be true, we are back to nonlocality, although it is not clear how this can be related to quantum nonlocality. But there may be some indications that such a relation exists \cite{Sto2020PostDeterminedBlockUniverse}.
\end{comment}

\begin{option}
\label{opt:nonlocality-in-space}
Theorem \ref{thm:nonlocality-O} is correct, and mental states are indeed O-nonlocal.
\end{option}
\begin{comment}
\label{comm:nonlocality-in-space}
I think this is the right conclusion. But I have no explanation for how this works, and no relevant understanding of the consequences of this option.
In Sec. \sref{s:quantum} I argued that there is a strong parallelism with quantum nonlocality, but this does not answer the questions, it merely provides a physical support for nonlocality.
\end{comment}


\begin{thebibliography}{10}

\bibitem{Aaronson2014PrettyHardProblemOfConsciousnessWhyIAmNotIIT}
S.~Aaronson.
\newblock Why {I} am not an {I}ntegrated {I}nformation {T}heorist (or, {T}he
  {U}nconscious {E}xpander), 2014.

\bibitem{DavidZAlbert2003TimeAndChance}
D.Z. Albert.
\newblock {\em Time and chance}.
\newblock Harvard University Press, Cambridge, Massachusetts; London, England,
  2003.

\bibitem{DavidZAlbert2015AfterPhysics}
D.Z. Albert.
\newblock {\em After physics}.
\newblock Harvard University Press, 2015.

\bibitem{adm2008admRepublication}
R.~Arnowitt, S.~Deser, and C.~W. Misner.
\newblock Republication of: The dynamics of general relativity.
\newblock {\em General Relativity and Gravitation}, 40(9):1997--2027, 2008.

\bibitem{BillNye2013ThinkingThatImThinking}
The Origins~Project at~ASU.
\newblock The great debate: {T}he storytelling of science (official) - ({P}art
  2/2).
\newblock
  \href{http://youtu.be/40YIIaF1qiw#t=13m50s}{http://youtu.be/40YIIaF1qiw\#t=13m50s}.

\bibitem{SEP-qt-consciousness}
H.~Atmanspacher.
\newblock Quantum approaches to consciousness.
\newblock In E.N. Zalta, editor, {\em The Stanford Encyclopedia of Philosophy}.
  Metaphysics Research Lab, Stanford University, winter 2019 edition, 2019.

\bibitem{Atmanspacher2019ContextualityRevisitedSignalingMayDifferFromCommunicating}
H.~Atmanspacher and T.~Filk.
\newblock Contextuality revisited: Signaling may differ from communicating.
\newblock In {\em Quanta and Mind}, pages 117--127. Springer, 2019.

\bibitem{Baker1958FormulationofQMPhaseSpace}
G.A. Baker~Jr.
\newblock Formulation of quantum mechanics based on the quasi-probability
  distribution induced on phase space.
\newblock {\em Phys. Rev.}, 109(6):2198, 1958.

\bibitem{BasievaEtAl2019TrueContextualityBeatsDirectInfluencesInHumanDecisionMaking}
I.~Basieva, V.H. Cervantes, E.N. Dzhafarov, and A.~Khrennikov.
\newblock True contextuality beats direct influences in human decision making.
\newblock {\em J. Exp. Psychol. Gen.}, 148:1925--1937, 2019.

\bibitem{Bekenstein1981UniversalUpperBound}
J.D. Bekenstein.
\newblock Universal upper bound on the entropy-to-energy ratio for bounded
  systems.
\newblock {\em Phys. Rev. D}, 23(2):287, 1981.

\bibitem{Bekenstein2005HowDoesEntropyInformationBoundWork}
J.D. Bekenstein.
\newblock How does the entropy/information bound work?
\newblock {\em Found. Phys.}, 35(11):1805--1823, 2005.

\bibitem{Bell2004SpeakableUnspeakable}
J.S. Bell.
\newblock {\em {S}peakable and unspeakable in quantum mechanics: {C}ollected
  papers on quantum philosophy}.
\newblock Cambridge University Press, 2004.

\bibitem{BlackmoreTroscianko2018ConsciousnessAnIntroduction}
S.~Blackmore and E.T. Tro{\'s}cianko.
\newblock {\em Consciousness: An Introduction}.
\newblock Routledge, 2018.

\bibitem{NedBlock1978TroublesWithFunctionalism}
N.~Block.
\newblock Troubles with functionalism.
\newblock 1978.

\bibitem{Bousso2002HolographicPrinciple}
R.~Bousso.
\newblock The holographic principle.
\newblock {\em Rev. Mod. Phys.}, 74(3):825, 2002.

\bibitem{SeanCarroll2010FromEternityToHere}
S.M. Carroll.
\newblock {\em From eternity to here: the quest for the ultimate theory of
  time}.
\newblock Penguin, 2010.

\bibitem{Casini2008RelativeEntropyAndTheBekensteinBound}
H.~Casini.
\newblock Relative entropy and the {B}ekenstein bound.
\newblock {\em Class. Quant. Grav.}, 25(20):205021, 2008.

\bibitem{CervantesDzhafarov2018SnowQueenIsEvilAndBeautifulExperimentalEvidenceProbabilisticContextualityInHumanChoices}
V.H. Cervantes and E.N. Dzhafarov.
\newblock Snow {Q}ueen is evil and beautiful: {E}xperimental evidence for
  probabilistic contextuality in human choices.
\newblock {\em Decision}, 5(3):193, 2018.

\bibitem{Chalmers1995HardProblem}
D.J. Chalmers.
\newblock Facing up to the problem of consciousness.
\newblock {\em Journal of consciousness studies}, 2(3):200--219, 1995.

\bibitem{Chalmers1996ConsciousMindSearchFundamentalTheory}
D.J. Chalmers.
\newblock {\em The conscious mind: In search of a fundamental theory}.
\newblock Oxford university press, 1996.

\bibitem{Churchland1981EliminativeMaterialism}
P.M. Churchland.
\newblock Eliminative materialism and propositional attitudes.
\newblock {\em The Journal of Philosophy}, 78(2):67--90, 1981.

\bibitem{ChrisClarke1995TheNonlocalityOfMind}
C.J.S. Clarke.
\newblock The nonlocality of mind.
\newblock {\em Journal of Consciousness Studies}, 2(3):231--240, 1995.

\bibitem{sep-chinese-room}
D.~Cole.
\newblock {The {C}hinese {R}oom Argument}.
\newblock In E.N. Zalta, editor, {\em The {Stanford} Encyclopedia of
  Philosophy}. Metaphysics Research Lab, Stanford University, spring 2020
  edition, 2020.

\bibitem{Damasio2012SelfComesToMind}
A.~Damasio.
\newblock {\em Self comes to mind: {C}onstructing the conscious brain}.
\newblock Vintage, 2012.

\bibitem{Dennett1993ConsciousnessExplained}
D.~Dennett.
\newblock {\em Consciousness explained}.
\newblock Penguin UK, 1993.

\bibitem{Dennett2016Illusionism}
D.~Dennett.
\newblock Illusionism as the obvious default theory of consciousness.
\newblock {\em J Conscious Stud}, 23(11-12):65--72, 2016.

\bibitem{Dennett2017FromBacteriaToBachAndBack}
D.~Dennett.
\newblock {\em From bacteria to {B}ach and back: {T}he evolution of minds}.
\newblock WW Norton \& Company, 2017.

\bibitem{Dretske1997NaturalizingTheMind}
F.~Dretske.
\newblock {\em Naturalizing the mind}.
\newblock MIT Press, 1997.

\bibitem{Eddington1928NatureOfThePhysicalWorld}
A.~Eddington.
\newblock {\em The nature of the physical world}.
\newblock Dent, London, 1928.

\bibitem{EngelEtAl2007EvidenceQuantumCoherencePhotosynthesis}
G.S. Engel, T.S. Calhoun, E.L. Read, T.-K. Ahn, T.~Man{\v{c}}al, Y.-C. Cheng,
  R.E. Blankenship, and G.R. Fleming.
\newblock Evidence for wavelike energy transfer through quantum coherence in
  photosynthetic systems.
\newblock {\em Nature}, 446(7137):782--786, 2007.

\bibitem{Feynman1965TheCharacterOfPhysicalLaw}
R.P. Feynman.
\newblock {\em The character of physical law}.
\newblock MIT press, 1965.

\bibitem{MatthewFisher2015QuantumCognitionProcessingWithNuclearSpinsInTheBrain}
M.P.A. Fisher.
\newblock Quantum cognition: The possibility of processing with nuclear spins
  in the brain.
\newblock {\em Ann. Phys.}, 362:593--602, 2015.

\bibitem{Fodor1975LanguageOfThought}
J.A. Fodor.
\newblock {\em The language of thought}, volume~5.
\newblock Harvard University Press, 1975.

\bibitem{Goff2017ConsciousnessAndFundamentalReality}
P.~Goff.
\newblock {\em Consciousness and fundamental reality}.
\newblock Oxford University Press, 2017.

\bibitem{HaeslerBeauregard2013NDENonlocalMind}
Trent-Von Haesler and Mario Beauregard.
\newblock Near-death experiences in cardiac arrest: {I}mplications for the
  concept of non-local mind.
\newblock {\em Archives of clinical psychiatry (S{\~a}o Paulo)},
  40(5):197--202, 2013.

\bibitem{HaganHameroffTuszynski2002QuantumComputationInBrainMicrotubulesDecoherenceAndBiologicalFEasibility}
S.~Hagan, S.R.~Hameroff R, and J.A. Tuszy{\'n}ski.
\newblock Quantum computation in brain microtubules: {D}ecoherence and
  biological feasibility.
\newblock {\em Phys. Rev. E}, 65(6):061901, 2002.

\bibitem{Hameroff2006ConsciousnessNeurobiologyAndQuantumMechanics}
S.~Hameroff.
\newblock Consciousness, neurobiology and quantum mechanics: {T}he case for a
  connection.
\newblock In {\em The emerging physics of consciousness}, pages 193--253.
  Springer, 2006.

\bibitem{HameroffPenrose2017OrchORUpdatedReview}
S.R. Hameroff and R.~Penrose.
\newblock Consciousness in the universe -- an updated review of the “{Orch
  OR}” theory.
\newblock In {\em Biophysics of Consciousness: {A} Foundational Approach},
  pages 517--599. World Scientific, 2017.

\bibitem{HartmannDurBriegel2006SteadyStateEntanglementInOpenNoisyQuantumSystems}
L.~Hartmann, W.~D{\"u}r, and H.J. Briegel.
\newblock Steady-state entanglement in open and noisy quantum systems.
\newblock {\em Phys. Rev. A}, 74(5):052304, 2006.

\bibitem{Heisenberg1958PhysicsAndPhilosophy}
W.~Heisenberg.
\newblock {\em Physics and Philosophy: {T}he Revolution in Modern Science}.
\newblock Harper \& Brothers Publishers, New York, 1958.

\bibitem{Hepp2012CoherenceAndDecoherenceInTheBrain}
K~Hepp.
\newblock Coherence and decoherence in the brain.
\newblock {\em Journal of mathematical physics}, 53(9):095222, 2012.

\bibitem{Hofstadter2007-I-am-a-strange-loop}
D.R. Hofstadter.
\newblock {\em I am a strange loop}.
\newblock Basic books, New York, 2007.

\bibitem{kevtris2013NANDputer}
K.~Horton.
\newblock {NANDputer lives!}, 2013.

\bibitem{WilliamJames1890ThePrinciplesOfPsychology}
W.~James.
\newblock {\em The principles of psychology}.
\newblock Benton; Chicago, IL, USA, 1890.

\bibitem{KochHepp2006QuantumMechanicsInTheBrain}
C.~Koch and K.~Hepp.
\newblock Quantum mechanics in the brain.
\newblock {\em Nature}, 440(7084):611--611, 2006.

\bibitem{LahavShanks1992HowToBeAScientificallyRespectablePropertyDualist}
R.~Lahav and N.~Shanks.
\newblock How to be a scientifically respectable" property-dualist".
\newblock {\em The Journal of Mind and Behavior}, pages 211--232, 1992.

\bibitem{sep-content-externalism}
J.~Lau and M.~Deutsch.
\newblock Externalism about mental content.
\newblock In E.N. Zalta, editor, {\em The {S}tanford Encyclopedia of
  Philosophy}. Metaphysics Research Lab, Stanford University, fall 2019
  edition, 2019.

\bibitem{Leibniz1989Monadology}
G.W. Leibniz.
\newblock The monadology.
\newblock In {\em Philosophical papers and letters}, pages 643--653. Springer,
  1989.

\bibitem{sep-functionalism}
J.~Levin.
\newblock Functionalism.
\newblock In E.N. Zalta, editor, {\em The {Stanford} Encyclopedia of
  Philosophy}. Metaphysics Research Lab, Stanford University, fall 2018
  edition, 2018.

\bibitem{LiParaoanu2009GenerationAndPropagationOfEntanglementInDrivenCoupledQubitSystems}
J.~Li and G.S. Paraoanu.
\newblock Generation and propagation of entanglement in driven coupled-qubit
  systems.
\newblock {\em New Journal of Physics}, 11(11):113020, 2009.

\bibitem{BarryLoewer2016MentaculusVision}
B.~Loewer.
\newblock {\em The mentaculus vision}.
\newblock Harvard University Press, 2020.

\bibitem{VanLommel2013NonLocalConsciousness}
Pim~Van Lommel.
\newblock Non-local consciousness, a concept based on scientific research on
  near-death experiences during cardiac arrest.
\newblock {\em Journal of Consciousness Studies}, 20(1-2):7--48, 2013.

\bibitem{LondonBauer1983TranslationOf1939QuantumMind}
F.~London and E.~Bauer.
\newblock The theory of observation in quantum mechanics.
\newblock In {\em Quantum theory and measurement}, pages 217--259. Princeton
  University Press, Princeton, NJ, 1983.

\bibitem{MccullochPitts1943LogicalCalculusOfTheIdeasImmanentInNervousActivity}
W.S. McCulloch and W.~Pitts.
\newblock A logical calculus of the ideas immanent in nervous activity.
\newblock {\em The Bulletin of Mathematical Biophysics}, 5(4):115--133, 1943.

\bibitem{Natterer2017ReadingWritingSingleAtomMagnets}
F.D. Natterer, K.~Yang, W.~Paul, P.~Willke, T.~Choi, T.~Greber, A.J. Heinrich,
  and C.P. Lutz.
\newblock Reading and writing single-atom magnets.
\newblock {\em Nature}, 543(7644):226--228, 2017.

\bibitem{JenniferOuellette2016NewSpinOnTheQuantumBrain}
Jennifer Ouellette.
\newblock A new spin on the quantum brain.
\newblock {\em Quanta Magazine}, 2016.

\bibitem{PanitchayangkoonEtAl2010LongLivedQuantumCoherencePhotosynthesis}
G.~Panitchayangkoon, D.~Hayes, K.A. Fransted, J.R. Caram, E.~Harel, J.~Wen,
  R.E. Blankenship, and G.S. Engel.
\newblock Long-lived quantum coherence in photosynthetic complexes at
  physiological temperature.
\newblock {\em Proceedings of the National Academy of Sciences},
  107(29):12766--12770, 2010.

\bibitem{sep-self-knowledge-externalism}
T.~Parent.
\newblock Externalism and self-knowledge.
\newblock In E.N. Zalta, editor, {\em The {Stanford} Encyclopedia of
  Philosophy}. Metaphysics Research Lab, Stanford University, fall 2017
  edition, 2017.

\bibitem{Penrose1989EmperorsNewMind}
R.~Penrose.
\newblock {\em The emperor's new mind}.
\newblock Oxford University Press, 1989.

\bibitem{Piccinini2004FunctionalismComputationalismAndMentalStates}
G.~Piccinini.
\newblock Functionalism, computationalism, and mental states.
\newblock {\em Stud. Hist. Philos. Sci. A}, 35(4):811--833, 2004.

\bibitem{PothosAndBusemeyer2013CanQuantumProbabilityNewDirectionCognitiveModeling}
E.M. Pothos and J.R. Busemeyer.
\newblock Can quantum probability provide a new direction for cognitive
  modeling?
\newblock {\em Behavioral and Brain Sciences}, 36(3):255--274, 2013.

\bibitem{Putnam1967Computationalism}
H.~Putnam.
\newblock Psychophysical predicates.
\newblock In W.~Capitan and D.~Merrill, editors, {\em Art, Mind, and Religion}.
  Pittsburgh: University of Pittsburgh Press, 1967.

\bibitem{Putnam1975TheMeaningOfMeaning}
H.~Putnam.
\newblock The meaning of ‘meaning’.
\newblock {\em Minnesota Studies in the Philosophy of Science}, 7:131--193,
  1975.

\bibitem{Pylkkanen2018QuantumTheoriesOfConsciousness}
P.~Pylkk{\"a}nen.
\newblock Quantum theories of consciousness.
\newblock In R.J. Gennaro, editor, {\em The {R}outledge Handbook of
  Consciousness}, pages 216--231. Routledge, 2018.

\bibitem{sep-eliminative-materialism}
William Ramsey.
\newblock Eliminative materialism.
\newblock In E.N. Zalta, editor, {\em The Stanford Encyclopedia of Philosophy}.
  Metaphysics Research Lab, Stanford University, spring 2019 edition, 2019.

\bibitem{sep-computational-mind}
M.~Rescorla.
\newblock {The Computational Theory of Mind}.
\newblock In Edward~N. Zalta, editor, {\em The {Stanford} Encyclopedia of
  Philosophy}. Metaphysics Research Lab, Stanford University, fall 2020
  edition, 2020.

\bibitem{RussellNorvig2016ArtificialIntelligence}
S.J. Russell and P.~Norvig.
\newblock {\em Artificial intelligence: a modern approach}.
\newblock Pearson Education Limited, Malaysia, 2016.

\bibitem{Sahu2013MultiLevelMemorySwitchingPropertiesSingleBrainMicrotubule}
S.~Sahu, S.~Ghosh, K.~Hirata, D.~Fujita, and A.~Bandyopadhyay.
\newblock Multi-level memory-switching properties of a single brain
  microtubule.
\newblock {\em Applied Physics Letters}, 102(12):123701, 2013.

\bibitem{Searle1980ChineseRoom}
J.R. Searle.
\newblock Minds, brains, and programs.
\newblock {\em Behavioral and brain sciences}, 3(3):417--424, 1980.

\bibitem{Searle1992RediscoveryOfMind}
J.R. Searle.
\newblock {\em The rediscovery of the mind}.
\newblock MIT press, Cambridge, Mass., 1992.

\bibitem{Stapp2004QuantumTheoryOfMindBrainInterface}
H.P. Stapp.
\newblock A quantum theory of the mind-brain interface.
\newblock In {\em Mind, matter and quantum mechanics}, pages 147--174.
  Springer, 2004.

\bibitem{Stapp2015QuantumMechanicsMindBrainConnection}
H.P. Stapp.
\newblock A quantum-mechanical theory of the mind-brain connection.
\newblock In {\em Beyond Physicalism}, pages 157--193. Rowman \& Littlefield
  Lanham, 2015.

\bibitem{Sto15c}
O.C. Stoica.
\newblock And the math will set you free.
\newblock {\em Foundational Questions Institute,
  \href{http://fqxi.org/community/essay/winners/2015.1}{``Trick or Truth: the
  Mysterious Connection Between Physics and Mathematics'' essay contest, third
  prize}}, 2015.
\newblock
  \href{http://fqxi.org/community/forum/topic/2383}{http://fqxi.org/community/forum/topic/2383},
  last accessed \today.

\bibitem{Stoica2017IndrasNetHolomorphicFundamentalness}
O.C. Stoica.
\newblock Indra's net -- {H}olomorphic fundamentalness.
\newblock {\em Foundational Questions Institute,
  \href{https://fqxi.org/community/forum/category/31426}{\emph{What Is
  ``Fundamental''} essay contest}}, 2017.
\newblock
  \href{https://fqxi.org/community/forum/topic/2971}{https://fqxi.org/community/forum/topic/2971},
  last accessed \today.

\bibitem{Sto2020NegativeWayToSentience}
O.C. Stoica.
\newblock The negative way to sentience.
\newblock {\em Preprint on
  \href{http://philsci-archive.pitt.edu/17036/}{http://philsci-archive.pitt.edu/17036/}
  - last accessed \today}, March 2020.

\bibitem{Sto2020PostDeterminedBlockUniverse}
O.C. Stoica.
\newblock The post-determined block universe.
\newblock {\em Quantum Stud. Math. Found.}, 2020.
\newblock \href{http://arxiv.org/abs/1903.07078}{arXiv:1903.07078}.

\bibitem{Sto2020StandardQuantumMechanicsWithoutObservers}
O.C. Stoica.
\newblock Standard quantum mechanics without observers.
\newblock {\em Preprint
  \href{http://arxiv.org/abs/2008.04930}{arXiv:2008.04930}}, 2020.

\bibitem{tegmark2000decoherenceBrain}
M.~Tegmark.
\newblock Importance of quantum decoherence in brain processes.
\newblock {\em Phys. Rev. E}, 61(4):4194, 2000.

\bibitem{Tegmark2014OurMathematicalUniverse}
M.~Tegmark.
\newblock {\em Our Mathematical Universe: {M}y Quest for the Ultimate Nature of
  Reality}.
\newblock Knopf Doubleday Publishing Group, 2014.

\bibitem{Tononi2016IntegratedInformationTheory}
G.~Tononi, M.~Boly, M.~Massimini, and Ch. Koch.
\newblock Integrated information theory: from consciousness to its physical
  substrate.
\newblock {\em Nature Reviews Neuroscience}, 17(7):450--461, 2016.

\bibitem{Turing1950ComputingAndIntelligence}
A.M. Turing.
\newblock Computing machinery and intelligence.
\newblock {\em Mind}, 59(236):433--460, 1950.

\bibitem{vonNeumann1955MathFoundationsQM}
J.~{von Neumann}.
\newblock {\em Mathematical Foundations of Quantum Mechanics}.
\newblock Princeton University Press, 1955.

\bibitem{vanWassenhove2009MindingTimeInAnAmodalRepresentationalSpace}
Virginie~Van Wassenhove.
\newblock Minding time in an amodal representational space.
\newblock {\em Phil. Trans. R. Soc. B}, 364(1525):1815--1830, 2009.

\bibitem{WeingartenDoraiswamyFisher2016NewSpinOnNeuralProcessingQuantumCognition}
C.P. Weingarten, P.M. Doraiswamy, and M.~Fisher.
\newblock A new spin on neural processing: quantum cognition.
\newblock {\em Frontiers in human neuroscience}, 10:541, 2016.

\bibitem{Wigner1932PhaseSpaceQM}
E.P. Wigner.
\newblock On the quantum correction for thermodynamic equilibrium.
\newblock {\em Phys. Rev.}, 40:749--759, Jun 1932.

\bibitem{Wigner1961WignerFriend}
E.P. Wigner.
\newblock {\em Remarks on the mind-body question}, pages 284--302.
\newblock Heinmann, London, 1961.

\end{thebibliography}

\end{document}